\documentclass[12pt,letterpaper]{article}
\pdfoutput=1

\usepackage{amsmath,amssymb,calc, amsthm,bbm, epsfig,psfrag, mathtools,comment}

\newtheorem{lem}{Lemma}[section]
\usepackage{dsfont}
\usepackage{mathrsfs}
\usepackage{esvect}
\usepackage{graphicx, enumerate, enumitem}
\usepackage{float}

\usepackage[numbers,sort&compress]{natbib}

\usepackage[dvipsnames]{xcolor}
\usepackage{tikz}
\usepackage{tikz-cd}
\usepackage{tensor}
\newcommand{\ul}{\underline}
\newcommand{\onehalf}{\frac{1}{2}}
\newcommand{\commutator}[2]{[#1,#2]}
\newcommand{\anticommutator}[2]{\{#1,#2\}}

\newcommand{\Lagrangian}{\mathcal{L}}

\usepackage[utf8]{inputenc} 
\definecolor{azure}{rgb}{0.0, 0.5, 1.0}
\definecolor{darkblue}{rgb}{0.15,0.35,0.7}
\definecolor{reddish}{rgb}{0.65, 0.2, 0.2}
\definecolor{brandeisblue}{rgb}{0.0, 0.44, 1.0}
\definecolor{ceruleanblue}{rgb}{0.16, 0.32, 0.75}
\definecolor{indigo(dye)}{rgb}{0.0, 0.25, 0.42}
\usepackage[linktocpage=true]{hyperref}
\hypersetup{
colorlinks=true,
citecolor=ceruleanblue,
linkcolor=ceruleanblue,
urlcolor=ceruleanblue,
pdfauthor={},
pdftitle={},
pdfsubject={}
}

\newcommand{\overbar}[1]{\mkern 1.5mu\overline{\mkern-1.5mu#1\mkern-1.5mu}\mkern 1.5mu}

\newcommand{\TTbar}{T\overbar{T}}
\newcommand{\TT}{T\overbar{T}}

\newcommand{\ov}[1]{\overbar{#1}}
\newcommand{\uu}[1]{\mathfrak u (#1)}

\newcommand{\ttr}[1]{{\rm tr}\left ( #1\right) }
\newcommand{\cL}{{\cal L}}

\newcommand{\fg}{\mathfrak g}
\newcommand{\vt}[1]{{v^{(2)}_{#1}}}
\newcommand{\jt}[1]{{j^{(2)}_{#1}}}
\newcommand{\jz}[1]{{j^{(0)}_{#1}}}

\DeclareMathOperator{\tr}{\text{tr}}

\DeclareFontEncoding{LS1}{}{}
\DeclareFontSubstitution{LS1}{stix}{m}{n}
\DeclareSymbolFont{stixsymbols}{LS1}{stixscr}{m}{n}
\SetSymbolFont{stixsymbols}{bold}{LS1}{stixscr}{b}{n}
\DeclareMathSymbol{\kay}{\mathalpha}{stixsymbols}{"6B}
\DeclareMathSymbol{\hay}{\mathalpha}{stixsymbols}{"68}

\makeatletter
\renewcommand\section{\@startsection {section}{1}{\z@}%
                               {-3.5ex \@plus -1ex \@minus -.2ex}
                               {2.3ex \@plus.2ex}%
                               {\normalfont\large\bfseries}}
\renewcommand\subsection{\@startsection{subsection}{2}{\z@}%
                                 {-3.25ex\@plus -1ex \@minus -.2ex}%
                                 {1.5ex \@plus .2ex}%
                                 {\normalfont\bfseries}}
\makeatother

\makeatletter
\newcommand*\bigcdot{\mathpalette\bigcdot@{.5}}
\newcommand*\bigcdot@[2]{\mathbin{\vcenter{\hbox{\scalebox{#2}{$\m@th#1\bullet$}}}}}
\makeatother

\newfont{\goth}{ygoth.tfm scaled 1200}                   

\numberwithin{equation}{section}

\newcommand{\del}{\partial}

\newcommand{\Cp}[1]{{\mathbb {C P}^{#1}}}

\def\d{\delta}

\def\l{\gamma}

\newcommand{\pa}{\partial}
\newcommand{\hf}{\frac12}

\newcommand{\bea}{\begin{eqnarray}}
\newcommand{\eea}{\end{eqnarray}}
\newcommand{\ba}{\begin{array}}
\newcommand{\ea}{\end{array}}

\def\double #1{#1{\hbox{\kern-2pt $#1$}}}

\newcommand{\bsubeq}{\begin{subequations}}
\newcommand{\esubeq}{\end{subequations}}

\def\tr{{\rm tr}}

\allowdisplaybreaks

\begin{document}
\begin{titlepage}
\begin{flushright}
\today
\end{flushright}
\vspace{5mm}

\begin{center}
{\Large \bf Soliton Surfaces and the Geometry of Integrable Deformations of the $\mathbb{CP}^{N-1}$ Model}
\end{center}

\begin{center}

{\bf
Christian Ferko${}^{a,b}$,
Michele Galli${}^{c}$,
Zejun Huang${}^{c}$,\\
Gabriele Tartaglino-Mazzucchelli${}^{c}$
} \\
\vspace{5mm}

\footnotesize{
${}^{a}$
{\it 
Department of Physics, Northeastern University, Boston, MA 02115, USA
}
 \\~\\
${}^{b}$
{\it 
The NSF Institute for Artificial Intelligence
and Fundamental Interactions
}
  \\~\\
${}^{c}$
{\it 
School of Mathematics and Physics, University of Queensland,
\\
 St Lucia, Brisbane, Queensland 4072, Australia}
}
\vspace{2mm}
~\\
\texttt{
c.ferko@northeastern.edu,
m.galli@uq.edu.au,
zejun.huang@uq.net.au,
g.tartaglino-mazzucchelli@uq.edu.au
}\\
\vspace{2mm}

\end{center}

\begin{abstract}
\baselineskip=14pt

The $\mathbb{CP}^{N-1}$ model is an analytically tractable $2d$ quantum field theory which shares several properties with $4d$ Yang-Mills theory. By virtue of its classical integrability, this model also admits a family of integrable higher-spin auxiliary field deformations, including the $\TTbar$ deformation as a special case. We study the $\mathbb{CP}^{N-1}$ model and its deformations from a geometrical perspective, constructing their soliton surfaces and recasting physical properties of these theories as statements about surface geometry. 
We examine how the $\TTbar$ flow affects the unit constraint in the $\mathbb{CP}^{N-1}$ model and prove that any solution of this theory with vanishing energy-momentum tensor remains a solution under analytic stress tensor deformations -- an argument that extends to generic dimensions and instanton-like solutions in stress tensor flows including the non-analytic, $2d$, root-${\TTbar}$ case and classes of higher-spin, Smirnov-Zamolodchikov-type, deformations.
Finally, we give two geometric interpretations for general $\TTbar$-like deformations of symmetric space sigma models, showing that such flows can be viewed as coupling the undeformed theory to a unit-determinant field-dependent metric, or using a particular choice of moving frame on the soliton surface.

\end{abstract}
\vspace{5mm}

\vfill
\end{titlepage}


\renewcommand{\thefootnote}{\arabic{footnote}}
\setcounter{footnote}{0}

\tableofcontents{}
\vspace{0.5cm}
\bigskip\hrule

\newpage
\section{Introduction}

Understanding the dynamics of quantum field theories at strong coupling remains one of the most important challenges in theoretical physics. Famously, $4d$ Yang-Mills (YM) theory is asymptotically free, and hence weakly coupled in the UV, but becomes strongly coupled in the IR. For years, this has motivated a search for analogue theories which share some properties with Yang-Mills theory but which are more amenable to theoretical analysis.

One well-known such analogue, or toy model, is the $\mathbb{CP}^{N-1}$ theory in two spacetime dimensions. Classically, this theory describes a collection of massless, interacting, complex scalar fields $\phi_I$. The $\mathbb{CP}^{N-1}$ model is analytically tractable in at least two senses:

\begin{enumerate}[label = (\roman*)]
    \item\label{classical} The theory is \emph{classically} integrable due to its realization as a symmetric space sigma model, which means that its equations of motion are equivalent to the flatness of a Lax connection. As a consequence, the theory possesses an infinite collection of conserved charges, which can be shown to mutually Poisson-commute. The existence of this tower of conserved quantities severely constrains the dynamics of the theory, and opens up the possibility of applying specialized techniques to study its behavior. Integrability is shared with $4d$ self-dual YM, but not with the full $4d$ YM theory.

    \item\label{quantum} Certain calculations can be performed for the \emph{quantum} $\mathbb{CP}^{N-1}$ model, especially at large $N$. Although the theory is classically conformally invariant, quantum effects generate a scale $\Lambda_{\mathbb{CP}^{N-1}}$ (analogous to $\Lambda_{\text{QCD}}$) by the process of dimensional transmutation. One can compute the beta function for a quantity $v^2$ that determines the ``size'' of the target $\mathbb{CP}^{N-1}$ space, finding that $v^2$ becomes smaller in the IR, which means that the theory is strongly coupled at low energies like Yang-Mills. The theory has a mass gap: the classically massless scalars $\phi_I$ develop a mass set by $\Lambda_{\mathbb{CP}^{N-1}}$ in the quantum theory. The model also exhibits confinement, as the fields $\phi_I$ confine into singlet ``mesons'' constructed from $\phi$ and $\phi^\dagger$, much like quark confinement in QCD.\footnote{Therefore, depending on one's perspective, one can view the $\mathbb{CP}^{N-1}$ theory either as a toy model for pure Yang-Mills (thinking of $\phi_I$ as like gluons, in that they are classically massless yet the theory has a mass gap) or for QCD (viewing $\phi_I$ as more similar to quarks, e.g. because they confine into ``mesons'').}
\end{enumerate}

Points \ref{classical} and \ref{quantum} are complementary, but in some sense orthogonal, because the classical integrability of the $\mathbb{CP}^{N-1}$ model is not believed to persist at the quantum level -- with the exception of $\mathbb{CP}^1$, which is the $O(3)$ model -- unless one couples to additional fermions \cite{GOLDSCHMIDT1980392,Abdalla:1980jt,Abdalla:1981yt,Gomes:1982qh,Komatsu:2019hgc}.

In this work, we will focus entirely on the classical aspects of the $\mathbb{CP}^{N-1}$ theory and related models in order to leverage the integrability of point \ref{classical}. Although this prevents us from investigating interesting quantum phenomena such as confinement, dimensional transmutation, and so on, there are already useful similarities between $\mathbb{CP}^{N-1}$ and Yang-Mills at the classical level. For instance, the $\mathbb{CP}^{N-1}$ model possesses instanton solutions, owing to the fact that $\pi_2 \left( \mathbb{CP}^{N-1} \right) = \mathbb{Z}$, which allows for topologically non-trivial field configurations in different winding sectors.\footnote{Here $\pi_n$ denotes the $n$-th homotopy group and the \emph{winding number} $Q \in \mathbb{Z}$ labels homotopy classes.} These solutions are close analogues of the instanton solutions of $4d$ Yang-Mills that can exhibit winding due to the fact that $\pi_3 ( G ) = \mathbb{Z}$ for any simple, compact Lie group $G$. Such instantonic field configurations in the $\mathbb{CP}^{N-1}$ model -- or, more precisely, its deformations -- will be one of the focuses of this work.

We mentioned in point \ref{classical} that integrable models permit the use of specialized techniques. One of these is a geometrical reformulation of the physics of these theories which involves an immersed surface. Over the past four decades, the study of $2d$ integrable models through various immersion methods has seen significant development \cite{Sym:1981an,Sym:1983px,solitonsurfacesandtheirapplications,Bobenko1994,https://doi.org/10.1002/sapm19969619,CIESLINSKI19971,bobenko2000painleve,helein2001constant,grundland2012soliton,grundland2016immersionformulassolitonsurfaces,Grundland2016generalizedsymmetryapproach,Grundland_2011integrableequations,Yesmakhanova_2019,Conti_2019,GLandolfi_2003,Baran_2010,talukdar2024fokaslenellsderivativenonlinearschrodinger,Bertrand_2016,Bertrand_2017}. The surfaces constructed via these immersions are referred to as soliton surfaces. Two-dimensional integrable models admit a rich geometric interpretation through the construction of soliton surfaces from their associated Lax pairs. In this framework, the zero-curvature condition, originally a compatibility condition for a pair of linear equations, serves as the starting point for defining such surfaces via an immersion into a Lie algebra. In the classical setting of surfaces embedded in $\mathbb{R}^{3}$, the first and second fundamental forms satisfying the Gauss–Mainardi–Codazzi (GMC) equations uniquely determine the surface up to rigid motions. This uniqueness ensures that the geometric data encoded in the Lax pair translates consistently into a well-defined surface, transforming the original system of partial differential equations into a geometric problem governed by the GMC equations. This correspondence was first formalized by Sym, who introduced what is now known as the Sym-Tafel formula, establishing a direct link between the zero-curvature condition and the geometry of soliton surfaces \cite{Sym:1981an,Sym:1983px,solitonsurfacesandtheirapplications,Bobenko1994}. As such, the soliton surface approach not only offers a powerful reformulation of integrable dynamics in geometric terms but also provides a concrete framework for visualizing and analyzing these systems from a geometric perspective.

Immersion methods have already been successfully applied to the standard $\mathbb{CP}^{N-1}$ model; see, for instance, \cite{grundland2006description,grundland2008conformally,grundland2008surfaces,grundland2012soliton}. However, given one useful toy model such as $\mathbb{CP}^{N-1}$, it is natural to attempt to generalize to a larger family of theories while retaining the useful properties which make the ``seed'' model tractable for analysis. In our setting, where the desirable property of point \ref{classical} is classical integrability, the task is therefore to find integrable deformations of the $\mathbb{CP}^{N-1}$ model and study them using techniques such as immersion methods. Carrying out this procedure is the primary purpose of the present work.

Specifically, we will focus on integrable deformations of the $\mathbb{CP}^{N-1}$ model which are constructed from conserved quantities. The most famous of these is the $\TT$ deformation \cite{Zamolodchikov:2004ce,Smirnov:2016lqw,Cavaglia:2016oda}, which is built from the conserved energy-momentum tensor and which preserves integrability when applied to an integrable seed theory. In recent years, many related deformations of $2d$ field theories have been proposed and studied, such as higher-spin Smirnov-Zamolodchikov (SZ) operators \cite{Smirnov:2016lqw}, the root-$\TT$ deformation \cite{Ferko:2022cix}, and the auxiliary field deformations of sigma models \cite{Ferko:2024ali} which include $\TT$, root-$\TT$, and SZ as special cases (for other results on auxiliary field deformations, see \cite{Bielli:2024ach,Bielli:2025uiv,Bielli:2024oif,Bielli:2024fnp,Bielli:2024khq,Fukushima:2024nxm,Cesaro:2024ipq}). We will study various aspects of all these integrable deformations, when applied to the $\mathbb{CP}^{N-1}$ model, with the aim of understanding their effects on the soliton surfaces which geometrize physical properties of the theory and on the instantonic solutions which this model shares with $4d$ Yang-Mills.

We begin by reviewing the basics of the $\mathbb{CP}^{N-1}$ model and soliton surfaces in Sections \ref{sec:CPN_review} and \ref{sec:generalities_soliton}, respectively. In Section \ref{STCPN}, we construct the soliton surfaces associated with the model and specialize to the $\mathbb{CP}^{1}$ case, where we compute the first and second fundamental forms as well as the Gaussian and mean curvatures.
Next, in Section~\ref{s:ttbcpn}, we show that the unit constraint and the topological chiral (instanton) solutions of the $\Cp {N-1}$ model are not affected by $T\ov T $ deformation. More generally, we find that holonomic constraints of any sigma model are preserved by the $T\ov T $ deformation. In particular, we show that the instanton solutions are preserved under the deformation. This result will later be extended in Section \ref{s:instpres1} to a general proof that any solution with a vanishing stress tensor is preserved in analytic $\TTbar$-like deformations of $d$-dimensional field theories, extending the $2d$ case, and the $4d$ Yang-Mills case studied in \cite{Ferko:2024yua}. In Section~\ref{s:cpnaux}, we reformulate the model as a coset theory and couple it to a set of auxiliary fields to obtain a higher-spin deformation of the $\mathbb{CP}^{N-1}$ theory, a framework that encompasses cases such as $\TTbar$-like and SZ deformations \cite{Ferko:2022cix,Bielli:2024oif,Bielli:2024ach,Bielli:2025uiv}. We then construct the soliton surfaces of the deformed model and recompute the corresponding geometric quantities in this setting, showing that the Gauss curvature is untouched by the deformation in the $\mathbb{CP}^{1}$ case. Within this formalism, it becomes evident that the instanton solutions are preserved under the deformation, and this result is extended to the $\sqrt{\TTbar}$ case. Finally, in Section \ref{sec:geometric}, we give geometric interpretations of auxiliary field deformations of symmetric space sigma models as either (a) coupling the undeformed theory to a unit-determinant field-dependent metric, or (b) enacting a non-trivial transformation on the tangent vectors to the soliton surface. Section \ref{sec:conclusion} summarizes our work and presents some future directions. We collect ancillary calculations in Appendices \ref{Expcheck} and \ref{ap:rootlag}.

\section{Review of the $\mathbb{CP}^{N-1}$ model and soliton surfaces}\label{ST2review}

In this section, we review some well-known background material concerning the $\mathbb{CP}^{N-1}$ model and soliton surfaces of $2d$ integrable models. We refer to the incomplete sampling of works \cite{divecchiaEffectiveCPN,Witten:1978bc,Actor:1980ey,10.1143/PTP.62.226,PhysRevD.73.065011,Bonanno:2018xtd} and references therein, or the pedagogical review in Chapter 7 of \cite{Tong:GaugeTheory}, for more information on the $\mathbb{CP}^{N-1}$ model. Likewise, to learn more about soliton surfaces, the interested reader might consult the original paper \cite{Sym:1981an} or the appendices in \cite{Conti_2019}. Again, we emphasize that here, and in the rest of the paper, our analysis will be entirely classical.

\subsection{Aspects of the $\mathbb{CP}^{N-1}$ model}\label{sec:CPN_review}

The $\mathbb{CP}^{N-1}$ model, serving as a $2d$ analogue of $4d$ Yang-Mills theory \cite{Yamazaki:2017ulc,Witten:1978bc,Actor:1980ey}, is described by the Lagrangian
\begin{equation}
    \Lagrangian=(\partial^{\mu}\phi^{\dagger})(\partial_{\mu}\phi)+\phi^{\dagger}(\partial^{\mu}\phi)\phi^{\dagger}(\partial_{\mu}\phi)
    \label{CPNLagrangian1},
\end{equation}
where
\begin{equation}\label{scalar_defn}
    \phi=
    \begin{bmatrix}
    \phi_1\\
    .\\
    .\\
    .\\
    \phi_{N}
    \end{bmatrix}
\end{equation}
is a collection of $N$ complex scalars satisfying the unit constraint
\begin{equation}
|\phi|^2=\phi^{\dagger}\phi=1
\label{unitconstraint}.    
\end{equation}
We will also refer to $\phi$ as a scalar multiplet. Besides a global $SU(N)$ symmetry, the system possesses a $U(1)$ local symmetry mediated by the pseudo-gauge field $A_{\mu}$. This makes it viable to define the covariant derivative
\begin{equation}
    D_{\mu}=\partial_{\mu}-igA_{\mu}
    \label{defcovariantderivative},
\end{equation}
where $g$ is the coupling constant and $A_{\mu}$ is given by
\begin{equation}
    A_{\mu}=\frac{i}{2g}\Big((\partial_{\mu}\phi^{\dagger})\phi-\phi^{\dagger}(\partial_{\mu}\phi)\Big)
    \label{gaugefieldhermitian},
\end{equation}
written in a manifestly Hermitian form. Since $A_{\mu}$ is not an independent physical field being constructed out of $\phi$ and $\phi^\dagger$,  we refer to $A_{\mu}$ as the pseudo-gauge field. With Eq. \eqref{defcovariantderivative}, the Lagrangian can be rewritten in a more compact form
\begin{equation}
    \Lagrangian=|D_{\mu}\phi|^2=(D^{\mu}\phi)^{\dagger}(D_{\mu}\phi)
    \label{CPNLagrangian2},
\end{equation}
accompanied by a useful identity
\begin{equation}
    (D_{\mu}\phi)^{\dagger}\phi=\phi^{\dagger}(D_{\mu}\phi)=0
    \label{covariantinnerproductidentity}.
\end{equation}
Taking into account the constraint in Eq. \eqref{unitconstraint}, the Lagrangian becomes
\begin{equation}
    \Lagrangian=(D^{\mu}\phi)^{\dagger}(D_{\mu}\phi)-\lambda(\phi^{\dagger}\phi-1)
    \label{CPNLagrangian3},
\end{equation}
where $\lambda$ is the Lagrange multiplier, and the corresponding equation of motion reads
\begin{equation}
    D^2\phi+|D_{\mu}\phi|^2\phi=0,
    \label{EOMCPN-1}
\end{equation}
with the associated energy-momentum tensor given by
\begin{equation}
    T^{\mu\nu}=(D^{\mu}\phi)^{\dagger}(D^{\nu}\phi)+(D^{\nu}\phi)^{\dagger}(D^{\mu}\phi)-g^{\mu\nu}(D^{\alpha}\phi)^{\dagger}(D_{\alpha}\phi)
    \label{energymomentumtensor1}.
\end{equation}

\subsubsection*{\ul{\it Topological charge and instanton}}

The concept of topological charge varies depending on the context. The most common instance is the winding number, a quantity labeling instanton solutions that are well-defined on a Euclidean manifold but not on a Lorentzian manifold.

Here, we provide a brief and intuitive recap of how a winding number is defined for a complex scalar field (multiplet). Instead of directly beginning with a scalar multiplet, which is the case for the $\mathbb{CP}^{N-1}$ model, we first consider the simpler case of a single complex scalar field, which possesses a well-defined phase directly related to the field itself. Classically, such a field $\psi(x)$ can be expressed as
\begin{equation}
    \psi(x)=\rho(x)e^{i\theta(x)}
    \label{psidef1},
\end{equation}
where $\rho(x)$ is the magnitude of the field and $\theta(x)$ is the real-valued phase factor. This leads to
\begin{equation}
     d\theta=(\partial_{\mu}\theta)dx^{\mu}=\frac{\psi^{\dagger}(\partial_{\mu}\psi)-(\partial_{\mu}\psi^{\dagger})\psi}{2i|\psi|^2}dx^{\mu},
\end{equation}
and hence the winding number is given by
\begin{equation}
    Q=\frac{1}{4\pi i}\int\frac{\psi^{\dagger}(\partial_{\mu}\psi)-(\partial_{\mu}\psi^{\dagger})\psi}{|\psi|^2}dx^{\mu}
    \label{Qdef0}.
\end{equation}
In any $d\geq2$, the expression above can be rewritten using Stokes' theorem. In $d=2$ it becomes
\begin{equation}
    Q=\frac{1}{2\pi}\int d^2x\epsilon^{\mu\nu}(\partial_{\mu}\partial_{\nu}\theta)=0,
    \label{Qdef1}
\end{equation}
where the integrand itself vanishes identically due to the contraction between symmetric and antisymmetric tensors, rendering it a trivial quantity in $d\geq2$.\footnote{Although Eq. \eqref{Qdef1} is specific to $d=2$, the vanishing topological charge due to Stokes' theorem holds for any $d\geq2$.} While this definition of the winding number is intuitive, it appears to be non-trivial only in the $1d$ case. However, this vanishing result is in fact a remarkable feature, far from trivial, and we now turn to its significance.

In contrast to $\psi$, which has a well-defined phase in Eq. \eqref{psidef1}, $\phi$, as the complex scalar multiplet in the $\mathbb{CP}^{N-1}$ model, lacks a globally well-defined common phase. As a result, replacing $\psi$ with $\phi$ in Eq. \eqref{Qdef0} does not make it vanish. Moreover, for a finite action in the model, we require that the Lagrangian vanishes as $|x|\rightarrow\infty$. This translates to
\begin{equation}
    \lim_{|x|\rightarrow\infty}D_{\mu}\phi\rightarrow0,
\end{equation}
and the field asymptotically takes the form
\begin{equation}
    \phi(x)\rightarrow\phi_{0}e^{i\Theta(x)}
    \label{fieldasymptoticalform},
\end{equation}
where $\phi_{0}$ is an arbitrary constant multiplet and $\Theta$ is a real-valued phase field (see \cite{151eee58736b41eaab5e72d7ffcd88a7} for more details). Although each component may carry an arbitrary phase far from infinity, Eq. \eqref{fieldasymptoticalform} implies that their phases differ only by a constant at infinity, making it not only have a well-defined phase $\Theta$, but also vanish at the boundary due to Stokes' theorem, thereby preventing the topological charge from diverging when Eq. \eqref{Qdef0} is extended to the case of a complex scalar multiplet. With the results above, the link between the phase and the pseudo-gauge field is
\begin{equation}
    \Rightarrow\partial_{\mu}\Theta:=gA_{\mu},
\end{equation}
and the topological charge is given by
\begin{equation}
    Q=\frac{i}{4\pi}\int dx^{\mu}\Big((\partial_{\mu}\phi^{\dagger})\phi-\phi^{\dagger}(\partial_{\mu}\phi)\Big),
\end{equation}
or equivalently
\begin{equation}
    Q=\frac{i}{2\pi}\int d^2x\epsilon^{\mu\nu}\partial_{\mu}\Big((\partial_{\nu}\phi^{\dagger})\phi\Big)
    \label{CPNtopochargestokes}
\end{equation}
using Stokes' theorem.

We can define the charge density $\rho$ and topological current $j^{\mu}$ through
\begin{subequations}
\begin{gather}
    Q=\int d^2x \, \rho , \\
    \rho=\partial_{\mu}j^{\mu},
\end{gather}
\end{subequations}
and extract them from Eq. \eqref{CPNtopochargestokes} as
\begin{equation}
    j^{\mu}=\frac{g}{2\pi}\epsilon^{\mu\nu}A_{\nu}=\frac{i}{2\pi}\epsilon^{\mu\nu}(\partial_{\nu}\phi^{\dagger})\phi
    \label{topologicalcurrent2}
\end{equation}
and
\begin{equation}
    \rho=\frac{i}{2\pi}\epsilon^{\mu\nu}\partial_{\mu}\Big((\partial_{\nu}\phi^{\dagger})\phi\Big)
    \label{topologicalchargedensity2}.
\end{equation}

The instanton solution arises from minimizing the action for a given topological charge. Consider the following inequality:
\begin{equation}
    |(D_{\mu}\phi)\pm i\epsilon_{\mu\nu}(D^{\nu}\phi)|^2\geq0
    \label{inequality1},
\end{equation}
which is equivalent to
\begin{equation}
    \mathcal{S}\geq2\pi|Q|
    \label{inequality4},
\end{equation}
with $\mathcal{S}=\int d^2x\mathcal{L}$ being the action principle associated to the Lagrangian $\mathcal{L}$ of equations \eqref{CPNLagrangian1} and \eqref{CPNLagrangian2}.
From Eq. \eqref{inequality1}, we know the bound in Eq. \eqref{inequality4} is saturated by field configurations obeying the (anti-)self-duality condition
\begin{equation}
    D^{\mu}\phi\pm i\epsilon^{\mu\nu}D_{\nu}\phi=0
    \label{InstantonDE},
\end{equation}
which corresponds to a special subset of solutions to the equations of motion in Eq. \eqref{EOMCPN-1}.
The $Q=\pm1$ solution to the instanton equation in Eq. \eqref{InstantonDE} takes the form
\begin{equation}
    \phi=\frac{\omega}{|\omega|},
\end{equation}
where $\omega$ is another complex scalar multiplet with arbitrary norm, given by
\begin{subequations}
\begin{gather}
    \omega=u\Omega+v(x^{\pm}-a^{\pm}) , \\
    |\omega|^2=(x-a)^2+\Omega^2.
\end{gather}
\end{subequations}
Here, $u$ and $v$ are two orthonormal vectors satisfying
\begin{equation}
    |u|^2=|v|^2=1 \, , \quad u^{\dagger}v=0,
\end{equation}
and $x^{\pm}$ and $a^{\pm}$ are the complexified coordinates defined as
\begin{equation}
    x^{\pm}=x^0\pm ix^1 \, , \quad a^{\pm}=a^0\pm ia^1
    \label{complexifiedcoordinate}.
\end{equation}
The parameter $a$ denotes the center of the topological charge, while $\Omega\in\mathbb{R}_{+}$ serves as a scale parameter.
Since the $\mathbb{CP}^{N-1}$ model serves as an analogue of $4d$ Yang-Mills, the pseudo-gauge field corresponding to the instanton solution above takes the BPST form
\begin{equation}
    A_{\mu}=\pm\frac{1}{g}\epsilon_{\mu\nu}\frac{x^{\nu}-a^{\nu}}{(x-a)^2+\Omega^2},
    \label{BPSTinstantonform}
\end{equation}
mirroring the $4d$ Yang-Mills case, as expected, where the 't Hooft symbol
is replaced by $\epsilon_{\mu\nu}$ in the $2d$ setting \cite{Actor:1980ey}.

It is convenient to switch to the complex coordinates defined in Eq. \eqref{complexifiedcoordinate}, which leads to the following identity:
\begin{equation}
    \commutator{D_{+}}{D_{-}}=(D_{+}\phi)^{\dagger}(D_{+}\phi)-(D_{-}\phi)^{\dagger}(D_{-}\phi)=\pi\rho.
\end{equation}
In addition, one has
\begin{equation}
    \partial_{-}D_{+}+\partial_{+}D_{-}+\pi\rho=0.
\end{equation}

In this coordinate system, the Lagrangian becomes
\begin{equation}
    \Lagrangian=2\Big((D_{+}\phi)^{\dagger}(D_{+}\phi)+(D_{-}\phi)^{\dagger}(D_{-}\phi)\Big),
\end{equation}
and the equation of motion takes the form
\begin{equation}
    \anticommutator{D_{+}}{D_{-}}\phi+\onehalf\Lagrangian\phi=0.
\end{equation}
The instanton equation in Eq. \eqref{InstantonDE} now reduces to the topological chiral equation
\begin{equation}
    D_{\pm}\phi=0
    \label{topologicalchiralequation}.
\end{equation}

\subsubsection*{\ul{\it Projector formalism and the Lax pair}}

In this section, we introduce an alternative formulation of the $\mathbb{CP}^{N-1}$ model based on the projector $P$, and then present the associated Lax pair.

To satisfy the unit constraint in Eq. \eqref{unitconstraint}, it is convenient to express the field as
\begin{equation}
    \phi=\frac{f}{|f|},
\end{equation}
where $f$ is an arbitrary complex scalar multiplet. This transformation effectively converts the problem of a constrained field into that of an unconstrained one. We then define the projector
\begin{equation}
    P=\phi\phi^{\dagger}=\frac{ff^{\dagger}}{|f|^2},
\end{equation}
which satisfies the following identities:
\begin{subequations}
\begin{gather}
    \tr(P)=1,\\
    P^{\dagger}=P^2=P.
\end{gather}
\end{subequations}
With this, the Lagrangian of the model can be rewritten in terms of $P$, rather than $\phi$, as
\begin{equation}
    \Lagrangian=\onehalf\tr\Big((\partial^{\mu}P)(\partial_{\mu}P)\Big).
    \label{LagrangianProjectorformalism}
\end{equation}
The equation of motion in terms of $f$ is
\begin{equation}
    (1-P)\Big((\partial^{\mu}\partial_{\mu}f)-2(\partial_{\mu}f)\frac{[f^{\dagger}(\partial^{\mu}f)]}{|f|^2}\Big)=0
\end{equation}
or equivalently, at the matrix level,
\begin{equation}
    \commutator{(\partial^{\mu}\partial_{\mu}P)}{P}=0.
\end{equation}
We define the following object:
\begin{equation}
    K_{\mu}=\commutator{(\partial_{\mu}P)}{P},
\end{equation}
which satisfies
\begin{equation}
    K_{\mu}\phi=D_{\mu}\phi,
\end{equation}
allowing us to express the Lagrangian and the charge density as
\begin{equation}
    \Lagrangian=-\onehalf\tr(K^{\mu}K_{\mu})
\end{equation}
and
\begin{equation}
    \rho=\frac{i}{4\pi}\epsilon^{\mu\nu}\tr\Big(P\commutator{K_{\mu}}{K_{\nu}}\Big),
\end{equation}
respectively.

Since $K_{\mu}$ takes the form of a commutator, it is traceless. Additionally, given the Hermiticity of $P$, $K_{\mu}$ is anti-Hermitian. These two properties together imply
\begin{equation}\label{K_in_sun}
    K_{\mu}\in\mathfrak{su}(n),
\end{equation}
when working in real coordinates.

In complex coordinates, the Lagrangian, charge density, and equation of motion become
\begin{equation}
    \Lagrangian=-2\tr(K_{+}K_{-}),
    \label{LagrangianinK}
\end{equation}
\begin{equation}
    \rho=\frac{1}{\pi}\tr\Big(P\commutator{K_{+}}{K_{-}}\Big)
    \label{densitycomplex},
\end{equation}
and
\begin{equation}
    \partial_{+}K_{-}+\partial_{-}K_{+}=0,
\end{equation}
respectively. The energy-momentum tensor, in these coordinates, becomes
\begin{subequations}
\begin{gather}
    T_{\pm\pm}=-\tr(K_{\pm}^2),\\
    T_{\pm\mp}=0.
\end{gather}
\end{subequations}
One can verify that the system admits the Lax pair \cite{Zakharov1978RelativisticallyIT,SALEEM20171080,grundland2012soliton}
\begin{equation}
    L_{\pm}=\frac{2}{1\pm z}K_{\pm},
    \label{Lax1}
\end{equation}
where $z\in\mathbb{C}$ is the spectral parameter, given the identity
\begin{equation}
    \commutator{K_{+}}{K_{-}}=\onehalf\Big(\partial_{+}K_{-}-\partial_{-}K_{+}\Big).
\end{equation}
The topological chiral (instanton) equation written in terms of the projector $P$ is
\begin{equation}
    P\partial_{\pm}P=0
    \label{instanconK}.
\end{equation}

\subsection{Generalities on soliton surfaces}\label{sec:generalities_soliton}

Many features of integrable $2d$ sigma models can be reinterpreted as geometrical statements about objects known as \emph{soliton surfaces} or \emph{Sym-Tafel surfaces}.\footnote{While the Sym-Tafel formula is a simple and effective method for constructing soliton surfaces, there exist other approaches such as the generalized Weierstrass and Fokas-Gel'fand formulas for immersion; see \cite{https://doi.org/10.1002/sapm19969619,grundland2012soliton,grundland2016immersionformulassolitonsurfaces,Grundland2016generalizedsymmetryapproach,Grundland_2011integrableequations,talukdar2024fokaslenellsderivativenonlinearschrodinger,Bertrand_2016,Bertrand_2017} for details.} In this subsection, we will review some facts about the Sym-Tafel construction in order to motivate our study of the soliton surfaces of the $\mathbb{CP}^{N-1}$ model later in this work. Because many of these results hold for a broader class of theories than the $\mathbb{CP}^{N-1}$ model, for now we will work more generally.

Consider a classically integrable $2d$ field theory, which means that its equations of motion are equivalent to the flatness of a Lax connection $L_\mu ( x , z )$,
\begin{align}\label{bad_sign_lax}
    \partial_\mu L_\nu - \partial_\nu L_\mu - [ L_\mu , L_\nu ] = 0 \, , 
\end{align}
which holds for every value of a spectral parameter $z \in \mathbb{C}$ and where $x$ denotes the coordinates on the $2d$ spacetime $\Sigma$. We note that equation (\ref{bad_sign_lax}) corresponds to the condition $d_{L} L = 0$ in conventions where the covariant exterior derivative is defined by $d_L X = d X - L \wedge X$, which differs by a sign from another common convention $\tilde{d}_L X = d + L \wedge X$. One can translate between the two conventions by reversing the sign of the Lax connection, defining $\mathfrak{L}_\mu = - L_\mu$, which then satisfies the on-shell flatness condition 
\begin{align}\label{good_sign_lax}
    \partial_\mu \mathfrak{L}_\nu - \partial_\nu \mathfrak{L}_\mu + [ \mathfrak{L}_\mu , \mathfrak{L}_\nu ] = 0 \, . 
\end{align}
We will use the convention (\ref{bad_sign_lax}) in this paper, although one must remember to convert by reversing the sign when comparing to other results such as \cite{Bielli:2024oif}.

Now assume that the Lax connection $L_\mu$ is valued in a Lie algebra $\mathfrak{g}$. Examples of such theories include (i) the principal chiral model (PCM), where $L_\pm$ is proportional to the Maurer-Cartan form $j_\pm = g^{-1} \partial_\pm g \in \mathfrak{g}$ in light-cone coordinates; (ii) several integrable deformations of the PCM, such as the PCM with Wess-Zumino term, the Yang-Baxter deformation, and the lambda deformation; (iii) the $\mathbb{CP}^{N-1}$ model, which one can see via the Lax connection (\ref{Lax1}) since $K_\mu \in \mathfrak{su} ( n )$ according to equation (\ref{K_in_sun}); and (iv) more general symmetric space sigma models (SSSMs).

Given such a Lax connection $L_\mu ( x, z )$, one can define a function $\Phi ( x, z )$ which is valued in $G$, the Lie group associated with the Lie algebra $\mathfrak{g}$, by solving the linear spectral problem
\begin{align}\label{aux_lin}
    L_\mu ( x, z ) \Phi ( x, z ) = \partial_\mu \Phi ( x, z ) \, .
\end{align}
We then consider a function $r$ defined by the Sym-Tafel formula,
\begin{align}\label{sym_tafel}
    r ( x, z ) = \Phi^{-1} ( x, z ) \partial_z \Phi ( x, z ) \, .
\end{align}
For each fixed choice of spectral parameter $z$, the function $r$ maps a spacetime point $x$ into the Lie algebra $\mathfrak{g}$. Let the tangent vectors to this surface be
\begin{align}\label{tangent_defn}
    r_\mu = \partial_\mu r \, ,
\end{align}
which can be straightforwardly computed,
\begin{align}\label{tangent_computation}
    r_\mu &= \partial_\mu \left( \Phi^{-1} \partial_z \Phi \right) \nonumber \\
    &= \left( - \Phi^{-1} ( \partial_\mu \Phi ) \Phi^{-1} \right) \partial_z \Phi + \Phi^{-1} \partial_{z} \partial_{\mu} \Phi \nonumber \\
    &= - \Phi^{-1} L_\mu \partial_z \Phi + \Phi^{-1} ( \partial_z \mathfrak{L}_\mu ) \Phi + \Phi^{-1} \mathfrak{L}_\mu \partial_z \Phi \nonumber \\
    &= \Phi^{-1} ( \partial_z L_\mu ) \Phi \, , 
\end{align}
where we have used the defining property (\ref{aux_lin}) to simplify derivatives of $\Phi$.

Since $r_\mu = \mathrm{Ad}_\Phi^{-1} ( \partial_z L_\mu )$, given any Ad-invariant bilinear form $\langle \, \cdot \, , \, \cdot \rangle$ on $\mathfrak{g}$, the pull-back $g_{\mu \nu}$ of this form to the worldsheet $\Sigma$ can be written as
\begin{align}\label{pullback_metric}
    g_{\mu \nu} = \langle r_\mu , r_\nu \rangle = \langle \partial_z L_\mu , \partial_z L_\nu \rangle \, .
\end{align}
For instance, a common choice of bilinear pairing for matrix Lie algebras is
\begin{align}\label{trace_pairing}
    \langle X, Y \rangle = - \tr ( X, Y ) \, .
\end{align}
Because the Lax connection depends on the fields of the model, at any fixed value of the spectral parameter $z$, the pullback (\ref{pullback_metric}) defines a particular field-dependent metric on the worldsheet $\Sigma$. We will see field-dependent metrics play another role, in the context of auxiliary field deformations, in Section \ref{sec:field_dep_metric}.

The advantage of this construction is that it furnishes us with a dictionary between integrability statements and geometrical statements. As a simple example, on any smooth manifold, mixed partial derivatives must commute, so one has
\begin{align}
    \partial_\mu r_\nu = \partial_\mu \partial_\nu r = \partial_\nu \partial_\mu r = \partial_\nu r_\mu \, .
\end{align}
We will sometimes write $r_{\mu \nu} = \partial_\mu r_\nu$. Let us compute these. For any object $A$, one has
\begin{align}
    \partial_\mu \left( \Phi^{-1} A \Phi \right) &= \left( - \Phi^{-1} ( \partial_\mu \Phi ) \Phi^{-1} \right) A \Phi + \Phi^{-1} ( \partial_\mu A ) \Phi + \Phi^{-1} A \partial_\mu \Phi \nonumber \\
    &= \Phi^{-1} ( \partial_\mu A ) \Phi + \Phi^{-1} A {L}_\mu \Phi - \left( \Phi^{-1} ( {L}_\mu \Phi ) \Phi^{-1} \right) A \Phi  \nonumber \\
    &= \Phi^{-1} ( \partial_\mu A ) \Phi + \Phi^{-1} [ A, {L}_\mu ] \Phi \, .
\end{align}
Therefore,
\begin{align}\label{mixed_second}
    \partial_\mu r_\nu &= \partial_\mu \left( \Phi^{-1} ( \partial_z {L}_\nu ) \Phi \right) \nonumber \\
    &= \Phi^{-1} \left( \partial_\mu \partial_z {L}_\nu \right) \Phi + \Phi^{-1} [ \partial_z {L}_\nu , {L}_\mu ] \Phi \, .
\end{align}
Demanding equality of mixed second derivatives gives
\begin{align}\label{mixed_second_equal}
    0 &= \partial_\mu r_\nu - \partial_\nu r_\mu \nonumber \\
    &= \Phi^{-1} \left( \partial_\mu \partial_z {L}_\nu \right) \Phi + \Phi^{-1} [ \partial_z {L}_\nu , {L}_\mu ] \Phi - \left( \Phi^{-1} \left( \partial_\nu \partial_z {L}_\mu \right) \Phi + \Phi^{-1} [ \partial_z {L}_\mu , {L}_- ] \Phi  \right) \nonumber \\
    &= \Phi^{-1} \left( \partial_z \left( \partial_\mu {L}_\nu - \partial_\nu {L}_\mu - [ {L}_\mu , {L}_\nu ] \right) \right) \Phi \, .
\end{align}
This condition is simply the $z$-derivative of the flatness condition of the Lax connection, in the sign conventions (\ref{bad_sign_lax}), which automatically holds on-shell if the theory is integrable. 

Equality of mixed second derivatives is one consistency condition for a surface, but there are others. A two-dimensional surface $S$ embedded in $\mathbb{R}^n$ for $n \geq 3$ must satisfy the full set of classical structure equations, namely the Gauss equation describing the extrinsic curvature, the Codazzi-Mainardi relation for the covariant variation of the second fundamental form, and (in higher codimension) the Ricci equation for the curvature of the normal bundle.\footnote{Conversely, given an induced metric $g_{ij}$ on $S$, a second fundamental form $\tensor{h}{^a_i_j}$ symmetric in $i$ and $j$, and a metric normal connection $\tensor{\omega}{_i_a^b}$ satisfying the Gauss, Codazzi-Mainardi, and Ricci equations, there always locally exists an immersion which realizes the data $(g, h, \omega)$, unique up to an ambient isometry.} The full set of classical structure equations is rather involved for $n > 3$, so in this work we will primarily restrict attention to the case $n = 3$ where the normal bundle is rank-$1$ and the equations simplify. However, we will mention in passing that for general Sym-Tafel immersions in $\mathbb{R}^n$ for $n \geq 3$, all of these consistency conditions are implied by the flatness of the Lax connection and thus emerge from the underlying integrability of the theory.

The Sym-Tafel construction also geometrizes other interesting features of integrable theories; we will mention one more. Suppose that the worldsheet $\Sigma$ has the topology of a spatial cylinder of length $L$. Recall that the monodromy matrix of the Lax connection is
\begin{align}
    M ( z ) = \overset{\longleftarrow}{\mathcal{P} \mathrm{exp}} \left( \int_0^L L_\sigma \, d \sigma \right) \, ,
\end{align}
where we choose coordinates $x^\mu = ( \tau , \sigma )$ so that $\sigma$ is the spatial coordinate, and $\mathcal{P}$ denotes the path-ordering operation. This quantity $M ( z )$ can be expanded around appropriate values of the spectral parameter $z$ in order to generate an infinite tower of non-local conserved charges. The logarithmic derivatives of this monodromy matrix are related to the periods of the Sym-Tafel surface, since by properties of the path-ordered exponential,
\begin{align}
    \oint \frac{\partial r}{\partial \sigma} \, d \sigma = \int_0^L \Phi^{-1} ( \partial_z L_\sigma ) \Phi \, d \sigma = M ( z )^{-1} \partial_z M ( z ) \, .
\end{align}
Therefore, geometric data about the soliton surface directly encodes the generating function of the infinite tower of conserved charges in the model.

\subsection{Soliton surfaces of the $\mathbb{CP}^{N-1}$ model}\label{STCPN}

We now specialize to the case of the $\mathbb{CP}^{N-1}$ model. Several previous studies have investigated immersions for this model and its generalizations. In \cite{grundland2006description,grundland2008conformally,grundland2008surfaces,grundland2012soliton} the authors constructed surfaces associated with the $\mathbb{CP}^{N-1}$ model using the generalized Weierstrass and the Fokas-Gel’fand approaches, which were extended to Grassmannian sigma models in \cite{grundland2005description,delisle2015geometry}. It was argued in \cite{grundland2009analytic} that the generalized Weierstrass method produces a surface which is equivalent to the Sym-Tafel immersion used in this work, so geometrical data agree between immersion methods. We will not undertake a detailed comparison of our discussion here with the results of these previous works, since our goal is primarily to present the properties of the Sym-Tafel surface for the undeformed $\mathbb{CP}^{N-1}$ model to continue setting our notation and to provide a base case to contrast with deformations of this model. Later, in Section~\ref{s:defsolaux}, we will combine this framework with the auxiliary field formalism to examine how higher-spin deformations affect the geometry of the $\mathbb{CP}^{N-1}$ soliton surfaces.

We follow the construction of Section \ref{sec:generalities_soliton}, focusing on the case $\Phi\in SU(N)$, $L_{\mu}\in\mathfrak{su}(N)$. With the trace (\ref{trace_pairing}) as our bilinear form, the metric tensor of the soliton surface is
\begin{equation}
    g_{\mu\nu}=-\tr\Big((\partial_{z}L_{\mu})(\partial_{z}L_{\nu})\Big)
    \label{solitonsurfacemetrictensor}.
\end{equation}
Substituting the $\Cp{N-1}$ Lax pair \eqref{Lax1} into \eqref{solitonsurfacemetrictensor} gives the following expression for the metric on the soliton surface:
\begin{subequations}\label{pullback_killing_cartan}
\begin{gather}
    g_{\pm\mp}=\frac{4}{(1-z^2)^2}\tr(K_{+}K_{-})=-\frac{2}{(1-z^2)^2}\Lagrangian,\\[0.4em]
    g_{\pm\pm}=-\frac{4}{(1\pm z)^4}\tr(K_{\pm}^2)=\frac{4}{(1\pm z)^4}T_{\pm\pm}.
\end{gather}
\end{subequations}
The two non-vanishing components of the energy-momentum tensor become proportional to their counterparts in the metric tensor, while the vanishing diagonal part is replaced by the Lagrangian. In this case, the essential physical quantities are encoded geometrically on the soliton surface. It is worth emphasizing that, even though the surface constructed from the Sym-Tafel formula is a soliton surface, it does not isolate all the topological chiral features from others, as the Lax pair corresponds to the equation of motion, not just the topological chiral (instanton) branch. Note that the energy-momentum tensor proves to be identically zero when evaluated on the instanton solutions, a property that we will use in later parts of the paper.

The geometry of soliton surfaces and their curvatures is simplest to study when the Lie algebra is $\mathfrak{su}(2)$ (the $\mathbb{CP}^{1}$ case), in which case the normal bundle to the surface has rank 1, making it possible to define normal vectors, cross products, and related geometric notions. Thus, for the rest of this section, we will focus on the $\mathbb{CP}^{1}$ case instead of the general $\mathbb{CP}^{N-1}$ for simplicity. In this case, the embedding (\ref{sym_tafel}) takes the form
\begin{equation}\label{start_restricting_here}
    r : \Sigma \longrightarrow \mathbb R^3, 
\end{equation}
so that $r$ is an $\mathfrak {su}(2)\cong \mathbb R^3$ valued matrix, and the metric tensor becomes the first fundamental form of the surface. Furthermore, for $2d$ surfaces embedded into $3d$, the classical structure equations described in the preceding subsection simplify. In this setting, the Gauss-Mainardi-Codazzi (GMC) equations relate second derivatives of $r $ with the Christoffel symbols $\Gamma^\mu _{\nu\rho}$ of $g$ and the second fundamental form $h_{\mu\nu}$, which encodes the geometry of the surface normal to the soliton surface, as
\begin{equation}
    \del _\mu \del _\nu \,r = \Gamma ^\sigma _{\mu \nu }\del _\sigma  r + h_{\mu \nu }\, \widehat n \, .
\end{equation}
Here $\widehat n$ is the normal vector 
\begin{equation}
    \widehat n = \frac {[\del _1 r, \del _2 r ]}{|[\del _1 r, \del _2 r ]|}.
\end{equation}
We have already computed the second derivatives $r_{\mu \nu} = \partial_\mu \partial_\nu r$ in equation (\ref{mixed_second}), and similarly the unit normal is
\begin{equation}
     \widehat n = \big |[\del _z L _+, \del _z L_- ]\big |^{-1}\, \Phi ^{-1}\, [\del _z L _+, \del _z L_- ]\, \Phi \, .
\end{equation}
The cross product in this case is given by the commutator, and thus the normal vector is
\begin{equation}
    n=r_{+}\times r_{-}=-\frac{4}{(1-z^2)^2}\Phi^{-1}\commutator{K_{+}}{K_{-}}\Phi.
\end{equation}
For a two-dimensional soliton surface embedded in $\mathbb{E}^3$, the second fundamental form is
\begin{equation}\label{second_fundamental}
    h_{\mu\nu}=r_{\mu\nu}\cdot\hat{n}=\frac{-\tr\Big((\partial_{z}\partial_{\nu}L_{\mu}+\commutator{\partial_{z}L_{\mu}}{L_{\nu}})\commutator{\partial_{z}L_{+}}{\partial_{z}L_{-}}\Big)}{\sqrt{-\tr(\commutator{\partial_{z}L_{+}}{\partial_{z}L_{-}}^2)}},
\end{equation}
and for the $\mathbb{CP}^{1}$ case, it becomes
\begin{subequations}
\begin{gather}
    h_{\pm\pm}=\frac{\mp2}{(1\pm z)^2}\frac{\tr(\commutator{K_{+}}{K_{-}}\partial_{\pm}K_{\pm})}{\sqrt{-\tr(\commutator{K_{+}}{K_{-}}^2)}},\\[0.4em]
    h_{\pm\mp}=\frac{2}{1-z^2}\sqrt{-\tr(\commutator{K_{+}}{K_{-}}^2)}.
\end{gather}
\end{subequations}
The Gaussian and mean curvatures are
\begin{equation}
    \begin{split}
   K&=\det(g^{-1}h)
   =\frac{(1-z^2)^2}{4}\frac{\tr(\commutator{K_{+}}{K_{-}}\partial_{+}K_{+})\tr(\commutator{K_{+}}{K_{-}}\partial_{-}K_{-})+\tr(\commutator{K_{+}}{K_{-}}^2)^2}{\tr(\commutator{K_{+}}{K_{-}}^2)\Big(\tr(K_{+}^2)\tr(K_{-}^2)-\tr(K_{+}K_{-})^2\Big)}\\
   &=-\frac{(1-z^2)^2}{2},\\
    H&=\onehalf\tr(g^{-1}h)\\
    &=\frac{(1-z^2)^2}{4}\frac{\frac{\tr(K_{-}^2)}{(1-z)^2}\tr(\commutator{K_{+}}{K_{-}}\partial_{+}K_{+})-\frac{\tr(K_{+}^2)}{(1+z)^2}\tr(\commutator{K_{+}}{K_{-}}\partial_{-}K_{-})+2\tr(K_{+}K_{-})\frac{\tr(\commutator{K_{+}}{K_{-}}^2)}{1-z^2}}{\sqrt{-\tr(\commutator{K_{+}}{K_{-}}^2)}\Big(\tr(K_{+}^2)\tr(K_{-}^2)-\tr(K_{+}K_{-})^2\Big)},
    \end{split}
    \label{Curvoriginal}
\end{equation}
where the mean curvature reduces to a constant,
\begin{equation}
    H=-\frac{i(z^2-1)}{\sqrt{2}}
\end{equation}
for configurations obeying the topological chiral condition \eqref{instanconK}, in agreement with \cite{grundland2012soliton}.

To conclude this section, we state the formula for the surface area:
\begin{equation}
    A=\int\sqrt{\det(g)}dx^{+}dx^{-}=\frac{4}{(1-z^2)^2}\int\sqrt{T_{++}T_{--}-\frac{\Lagrangian^2}{4}}dx^{+}dx^{-},
\end{equation}
where the stress tensor and the Lagrangian appear through the metric tensor $g$.

Everything we discussed until this point concerns the undeformed $\Cp {N-1}$ model. Our goal will be to see how integrable higher spin deformations affect the physics of this theory, and this is what we will focus on from now on. We begin in the next section by analysing the simplest deformation of this type: $T\ov T$.

\section{$T\overbar{T}$ deformation of the $\mathbb{CP}^{N-1}$ model}\label{s:ttbcpn}

The $\TT$ deformation \cite{Zamolodchikov:2004ce,Smirnov:2016lqw,Cavaglia:2016oda} is a universal, solvable deformation of any two-dimensional quantum field theory. By ``universal'' we mean that it can be defined, even quantum-mechanically, in any $2d$ QFT which enjoys translation invariance, by virtue of the fact that the coincident point limit
\begin{align}
    \mathcal{O}_{\TT} ( x ) = \lim_{y \to x} \left( T^{\mu \nu} ( x ) T_{\mu \nu} ( y ) - \tensor{T}{^\mu_\mu} ( x ) \tensor{T}{^\nu_\nu} ( y ) \right)
\end{align}
gives rise to a well-defined local operator, up to total derivative terms that can be ignored when integrated over spacetime. By ``solvable'' we mean that deforming a seed theory by the integrated operator $\mathcal{O}_{\TT}$ according to the flow equation 
\begin{align}\label{TT_flow_equation}
    \frac{\partial S ( \gamma )}{\partial \gamma} = \frac{1}{2} \int d^2 x \, \mathcal{O}_{\TT}^{(\gamma)} \, ,
\end{align}
where $\mathcal{O}_{\TT}^{(\gamma)}$ is constructed from $S ( \gamma )$, produces a modified theory for which many properties can be related to those of the undeformed theory. In particular, and most relevant for our purposes, the $\TT$ deformation preserves integrability when applied to an integrable seed theory \cite{Smirnov:2016lqw}. At the classical level, and restricting to certain integrable sigma models, this fact can be understood by realizing the $\TT$ flow as a special case of certain auxiliary field deformations to be explained in Section \ref{ss:auxrev}. Rather than presenting an exhaustive overview of the $\TT$ deformation, we refer the reader to the reviews \cite{Jiang:2019epa,He:2025ppz} for further details.

We will be interested in the classical $T\ov T$-deformed $\mathbb{CP}^{N-1}$ model. There are several equivalent ways to define the seed theory. If one parameterizes the scalars $\phi$ of (\ref{scalar_defn}) directly using coordinates $X^m$ on $\mathbb{CP}^{N-1}$, the theory can be presented in sigma model form
\begin{align}\label{sigma_model_form}
    S = \int d^2 x \, g_{mn} ( X ) \partial_\mu X^m \partial^\mu X^n \, ,
\end{align}
where $g_{mn} ( X )$ is the metric on $\mathbb{CP}^{N-1}$. The $\TT$ deformation of an action (\ref{sigma_model_form}) with generic metric follows from the analysis of \cite{Cavaglia:2016oda} and involves a square root. Because the degrees of freedom $X^m$ in the deformed model still parameterize points on $\mathbb{CP}^{N-1}$ by construction, the unit constraint (\ref{unitconstraint}) is manifestly undeformed in this presentation.

A different way to define the undeformed $\mathbb{CP}^{N-1}$ model is using the fields $\phi$ with Lagrangian (\ref{CPNLagrangian1}), supplemented by a Lagrange multiplier term which enforces this unit constraint (\ref{unitconstraint}). Viewing this Lagrange multiplier term as a potential, the $\TT$ deformation of the theory in this formulation is also of the general form considered in \cite{Cavaglia:2016oda}. Because this presentation is equivalent to the definition (\ref{sigma_model_form}), it is natural to expect the $\TT$ deformations of the two seed theories to agree, and in particular that the unit constraint also remains unmodified along this flow. However, we find it worthwhile to verify this expectation explicitly, since theories which are equivalent in flat spacetime may nonetheless couple differently to the metric and therefore have different-looking stress tensors and $\TT$ flows. As a simple example, a free massless boson in $2d$ is equivalent to a free massless Dirac fermion (by the process of bosonization/fermionization) but the $\TT$ deformation of a free boson gives the Nambu-Goto action while the classical $\TT$ deformation of a free fermion merely introduces a four-fermion interaction, and these deformed theories appear rather different.

We will therefore directly study the $\TT$-deformed $\mathbb{CP}^{N-1}$ model in the $\phi$ variables and examine whether the unit constraint (\ref{unitconstraint}) can be modified along the flow, finding that it cannot. Because the $\TT$-deformed action for a collection of scalars involves a square root, we will also verify that the argument of this square root remains non-negative for topological chiral solutions of our interest, which serves as a preliminary check that these field configurations have well-defined actions in the deformed theory.

\subsection{Deformation of the Lagrangian}

We will write the defining flow equation (\ref{TT_flow_equation}) as
\begin{equation}
    \frac{\partial\Lagrangian(\gamma)}{\partial\gamma}=\TTbar(\gamma)
    \label{TTbardeformation},
\end{equation}
where $\TTbar(\gamma)$, constructed from the energy-momentum tensor, is the scalar that drives the deformation, and $\gamma$ is the deformation parameter.
Since the full theory depends on $\gamma$, we adopt the shorthand convention that when $\gamma=0$, the parameter is omitted from all tensors. For example,
\begin{equation}
    T^{\mu}_{\ \nu}=T^{\mu}_{\ \nu}(\gamma=0)
\end{equation}
denotes the energy-momentum tensor of the seed (undeformed) system.
Given an undeformed system, we express $\TTbar(\gamma)$ using the method in \cite{Ferko:2022cix}, and then solve the differential equation in Eq. \eqref{TTbardeformation} to obtain the deformed Lagrangian. In $2d$, the $\TTbar$ operator is defined as the determinant of the energy-momentum tensor and takes the form
\begin{equation}
    \TTbar:=\det(T_{\mu\nu})=\onehalf\epsilon^{\mu\nu}\epsilon^{\alpha\beta}T_{\mu\alpha}T_{\nu\beta},
\end{equation}
which can be rewritten as
\begin{equation}
    \TTbar=\onehalf\Big(T^{\mu}_{\ \nu}T^{\nu}_{\ \mu}-(T^{\mu}_{\ \mu})^2\Big).
    \label{TMinkowskispace}
\end{equation}
Define
\begin{equation}
    X^{\mu\nu}:=(D^{\mu}\phi)^{\dagger}(D^{\nu}\phi)+(D^{\nu}\phi)^{\dagger}(D^{\mu}\phi),
\end{equation}
which is symmetric and real. From this, we introduce two scalar quantities:
\begin{subequations}
\begin{gather}
    X:=X^{\mu\nu}X_{\mu\nu},\\
    Y:= \tensor{X}{^\mu_\mu} =2(D^{\mu}\phi)^{\dagger}(D_{\mu}\phi)=2\Lagrangian.
\end{gather}
\end{subequations}
With these definitions, the energy-momentum tensor given in Eq. \eqref{energymomentumtensor1} can be expressed as
\begin{equation}
    T^{\mu\nu}=X^{\mu\nu}-g^{\mu\nu}\frac{Y}{2}.
\end{equation}
Since $X$ and $Y$ are the only scalar quantities that can be constructed from the energy-momentum tensor, they constitute the only possible contributions to $\TTbar(\gamma)$ for any value of $\gamma$. This leads to
\begin{equation}
    T^{\mu}_{\ \nu}(\gamma)=4X^{\mu\theta}X_{\theta\nu}\Lagrangian_{X}(\gamma)+2X^{\mu}_{\ \nu}\Lagrangian_{Y}(\gamma)-\delta^{\mu}_{\nu}\Lagrangian(\gamma),
    \label{Noetherenergymomentumtensorgamma}
\end{equation}
where
\begin{equation}
    \Lagrangian_{X}(\gamma):=\partial_{X}\Lagrangian(\gamma),\quad\Lagrangian_{Y}(\gamma):=\partial_{Y}\Lagrangian(\gamma).
\end{equation}
It is straightforward to verify that Eq. \eqref{Noetherenergymomentumtensorgamma} reduces to its counterpart in the seed theory when $\gamma=0$. In this case, we have $\Lagrangian_{X}(\gamma=0)=0$ and $\Lagrangian_{Y}(\gamma=0)=1$. Using Eq. \eqref{Noetherenergymomentumtensorgamma}, we compute $\TTbar(\gamma)$, which takes the form:
\begin{equation}
    \TTbar(\gamma)=-4(X-Y^2)^2\Lagrangian_{X}^2+2(X-Y^2)\Lagrangian_{Y}^2+4Y(X-Y^2)\Lagrangian_{X}\Lagrangian_{Y}+4X\Lagrangian_{X}\Lagrangian+2Y\Lagrangian_{Y}\Lagrangian-\Lagrangian^2,
\end{equation}
where all the occurrences of $\Lagrangian$ and its derivatives are understood to be functions of $\gamma$, which we suppress for notational convenience.

The deformed Lagrangian, as the solution to the partial differential equation in Eq. \eqref{TTbardeformation}, takes the form
\begin{equation}
    \Lagrangian(X,Y,\gamma)=\frac{1-\sqrt{1-2\gamma Y+2\gamma^2(Y^2-X)}}{2\gamma}
    \label{deformedLagrangian1}.
\end{equation}
If the constraint is embedded into the Lagrangian, then the deformed system takes the modified form
\begin{equation}
    \Lagrangian(X,Y,K,\gamma)=\frac{1-K\gamma-\sqrt{1-2\gamma(1-\frac{K\gamma}{2})Y+2\gamma^2(Y^2-X)(1-\frac{K\gamma}{2})^2}}{2\gamma(1-\frac{K\gamma}{2})}
    \label{deformedLagrangian2},
\end{equation}
where
\begin{equation}
    K=2\lambda(\phi^{\dagger}\phi-1),
\end{equation}
where $\lambda$ is the Lagrange multiplier.

Having reviewed the construction of the integrated, $T\ov T$-deformed Lagrangian for the $\Cp {N-1}$ model, we will now analyse the interplay between the unit constraint and the flow. 

\subsection{Deformation of the unit constraint}
In this section, we apply Picard iteration to examine how the deformation modifies the unit constraint in Eq. \eqref{unitconstraint}. For the undeformed Lagrangian with the unit constraint embedded, the $\TTbar$ operator takes the form
\begin{equation}
    \TTbar=\frac{X}{2}-\frac{Y^2}{4}-\frac{K^2}{4},
\end{equation}
which leads to the first-order corrected Lagrangian:
\begin{equation}
    \Lagrangian_{1}(X,Y,K,\gamma)=\frac{Y}{2}-\frac{K}{2}+\gamma\Big(\frac{X}{2}-\frac{Y^2}{4}-\frac{K^2}{4}\Big)
    \label{Lagrangian1stordercorrection}.
\end{equation}
At this order, there is no additional term linear in $K$, which implies that no new contribution arises to leading order in $\gamma$. This indicates the constraint remains unchanged at first order.

We then recompute the $\TTbar$ operator using the corrected Lagrangian (\ref{Lagrangian1stordercorrection}) and obtain
\begin{equation}
\begin{split}
    \TTbar_{1}(\gamma)&=-\frac{K^4\gamma^2}{16}-\frac{K^3\gamma}{4}-\frac{K^2X\gamma^2}{4}+\frac{K^2Y^2\gamma^2}{8}-\frac{K^2}{4}-\frac{KX\gamma}{2}\\
    &+\frac{KY^2\gamma}{4}-\frac{X^2\gamma^2}{4}+\frac{3XY^2\gamma^2}{4}+XY\gamma+\frac{X}{2}-\frac{5Y^4\gamma^2}{16}-\frac{Y^3 \gamma}{2}-\frac{Y^2}{4},
\end{split}    
\end{equation}
which gives the second-order corrected Lagrangian:
\begin{equation}
\begin{split}
    \Lagrangian_2(\gamma)&=\frac{Y}{2} - \frac{K}{2} + \gamma \Big(\frac{X}{2} - \frac{Y^2}{4} - \frac{K^2}{4}\Big) + \gamma^2\Big(-\frac{1}{8} K^3-\frac{1}{4} K X + \frac{1}{8} K Y^2 + \frac{1}{2} X Y - \frac{1}{4} Y^3\Big)\\
    &+\gamma^3\Big(-\frac{1}{48} K^4 - \frac{1}{12} K^2 X + \frac{1}{24} K^2 Y^2 - \frac{1}{12} X^2 + \frac{1}{4} X Y^2 - \frac{13}{128} Y^4\Big).
\end{split}
\end{equation}
At this order, new terms linear in $K$ appear compared to the original Lagrangian, indicating a possible modification of the constraint at second order in the deformation.

Gathering all terms linear in $K$, we obtain
\begin{equation}
    K \left( -\onehalf-\gamma^2\frac{X}{4}+\gamma^2\frac{Y^2}{8} \right) =\lambda(\phi^{\dagger}\phi-1) \left( -1-\gamma^2\frac{X}{2}+\gamma^2\frac{Y^2}{4} \right)
    \label{newconstraint1}.
\end{equation}
Since the unit constraint is factored out in Eq. \eqref{newconstraint1}, we are free to retain it if desired. However, the remaining factor could be interpreted as modifying the original constraint, suggesting the possibility of adopting a new constraint of the form
\begin{equation}
    \gamma^2\Big(\frac{Y^2}{4}-\frac{X}{2}\Big)=1.
    \label{newconstraint2}
\end{equation}
It turns out that the expression inside the parentheses is equal to $-\TTbar(\gamma=0)$, implying that Eq. \eqref{newconstraint2} would  force $\TTbar\rightarrow-\infty$ as $\gamma\rightarrow0$. This contradicts the seed theory and therefore indicates that this new constraint should not be adopted. Therefore, we conclude that the unit constraint is preserved under the deformation, with no enhanced constraints introduced, as expected from the deformation of the equivalent seed theory (\ref{sigma_model_form}).

While we have shown here that holonomic constraints cannot be enhanced for any sigma model, it would be interesting to explore how constrained Lagrangians of other classes of theories are deformed by $\TTbar$.
\subsection{Features of the deformed system}
We begin this section by recalling the Lagrangian of the deformed system given in Eq. \eqref{deformedLagrangian1},
which contains a square root term of the form
\begin{equation}\label{square_root_term}
    \sqrt{1-2\gamma Y+2\gamma^2(Y^2-X)}.
\end{equation}
A natural question that arises is whether this square root is always well-defined, that is, whether the expression under the square root remains non-negative for all physically relevant configurations.

In  Euclidean signature, and in
complex coordinates defined in Eq. \eqref{complexifiedcoordinate}, let $A=D_{+}\phi$ and $B=D_{-}\phi$. By the Cauchy-Schwarz inequality, $|A^{\dagger}B|^2\leq|A|^2|B|^2$, we obtain
\begin{equation}
    Y^2-X=8(|A|^2+|B|^2)^2-32|A^{\dagger}B|^2\geq8(|A|^2-|B|^2)^2\geq0.
\end{equation}
It then follows that
\begin{equation}
    1-2\gamma Y+2\gamma^2(Y^2-X)\geq-64\gamma^2|A^{\dagger}B|^2=-16\gamma^2T_{++}T_{--}
\end{equation}
in general, and
\begin{equation}
    1-2\gamma Y+2\gamma^2(Y^2-X)\geq0
\end{equation}
for topological chiral solutions where the undeformed energy-momentum tensor is zero. This demonstrates that, although the deformed Lagrangian is not guaranteed to be well-defined for arbitrary configurations of $\phi$, for topological chiral solutions it remains well-defined and, as we will explain later in this section, reduces to the undeformed theory.

The equation of motion and energy-momentum tensor of the deformed system are
\begin{equation}
\begin{split}
    &\partial_{\mu}\Big(\frac{2\gamma X^{\mu\nu}(D_{\nu}\phi_{i})+(1-2\gamma Y)(D^{\mu}\phi_{i})}{\sqrt{1-2\gamma Y+2\gamma^2(Y^2-X)}}\Big)-\frac{2\gamma X^{\mu\nu}[\phi^{\dagger}(\partial_{\nu}\phi)](D_{\mu}\phi_{i})+(1-2\gamma Y )[\phi^{\dagger}(\partial^{\mu}\phi)](D_{\mu}\phi_{i})}{\sqrt{1-2\gamma Y+2\gamma^2(Y^2-X)}}\\
    &-\frac{2\gamma X^{\mu\nu}\phi^{\dagger}(\partial_{\mu}D_{\nu}\phi)+(1-2\gamma Y)\phi^{\dagger}(\partial_{\mu}D^{\mu}\phi)}{\sqrt{1-2\gamma Y+2\gamma^2(Y^2-X)}}\phi_{i}=0
\end{split}
\label{deformedEOM}
\end{equation}
and
\begin{equation}
    T^{\mu}_{\ \nu}=\frac{2\gamma X^{\mu\theta}X_{\theta\nu}+(1-2\gamma Y)X^{\mu}_{\ \nu}}{\sqrt{1-2\gamma Y+2\gamma^2(Y^2-X)}}-\delta^{\mu}_{\nu}\frac{1-\sqrt{1-2\gamma Y+2\gamma^2(Y^2-X)}}{2\gamma}
    \label{deformedenergymomentumtensor},
\end{equation}
respectively.

Although the equation of motion appears complicated, one can verify that the topological chiral solutions remain valid when substituted into it. This validity is further supported by the behavior of the energy-momentum tensor, which vanishes entirely for the topological chiral solutions (see Appendix \ref{Expcheck} for the full computation). In \cite{Conti_2019}, the $\TTbar$ deformation was shown to be a field-dependent local coordinate transformation of the form
\begin{equation}\label{TT_field_dep_diff}
    \phi^{(\gamma)}(x)=\phi^{(0)}(y(x)),
\end{equation}
where
\begin{equation}
    dy^{\mu}=\Big(\delta^{\mu}_{\ \nu}+\gamma(\tilde{T}^{(\gamma)})^{\mu}_{\ \nu}(x)\Big)dx^{\nu},
\end{equation}
with
\begin{equation}
    (\tilde{T}^{(\gamma)})^{\mu}_{\ \nu}=-\epsilon^{\mu}_{\ \rho}\epsilon^{\sigma}_{\ \nu}(T^{(\gamma)})^{\rho}_{\ \sigma}.
\end{equation}
If the energy-momentum tensor vanishes, the coordinate transformation $dy^\mu = dx^\mu$ is at most a constant shift $y^\mu = x^\mu + a^\mu$, and by translation invariance of the undeformed theory (which we always assume, so that the seed theory has a conserved stress tensor), the original solution remains a solution in the deformed model.

Although the deformed Lagrangian (\ref{deformedLagrangian1}) satisfies the $\TT$ flow equation to all orders in $\gamma$ -- and not merely to some finite order in a perturbative expansion -- we will be interested in studying physical properties of the deformed theory which admit smooth limits to the undeformed theory. Therefore, we will impose one constraint on allowed values of the deformation parameter $\gamma$. For topological chiral solutions, one has
\begin{equation}
    1-2\gamma Y+2\gamma^2(Y^2-X)=(1-\gamma Y)^2,
\end{equation}
and to ensure the deformed Lagrangian reduces to the undeformed one, we require
\begin{equation}
    0\leq\gamma\leq\frac{1}{Y}.
\end{equation}
The $\mathbb{CP}^{N-1}$ model serves as a $2d$ analogue of $4d$ Yang–Mills theory, as previously discussed. The fact that the energy-momentum tensor vanishes for topological chiral solutions ensures that these solutions are preserved under the $\TTbar$ deformation, effectively leaving the system undeformed in this sector. This naturally raises the question of whether such a feature extends to other systems in $4d$, or even in arbitrary dimensions, provided they admit solutions with vanishing energy-momentum tensors. The answer is yes, and we will present a general proof in the next section. This extends the result first given in \cite{Ferko:2024yua} for the $4d$ Yang-Mills instantons and deformations quadratic in the energy-momentum tensor.

\section{Theorem about non-deformation of instantons}\label{s:instpres1}
In this section, we will show using quite general arguments that solutions with a trivial energy-momentum tensor are preserved by $T\ov T $-like deformations under mild analyticity conditions on the deformation. More specifically, let us consider a theory with Lagrangian $\cL^{(\gamma)}$ satisfying the following flow
\bea
\frac{\partial\cL^{(\gamma)}}{\partial\gamma}=f(T_{\mu\nu})
\label{appendix_proof_flow}
~,
\eea
where $f$ is analytic in $T_{\mu \nu}$, and $\gamma$ is the coupling associated with these kinds of flows. 
As anticipated, we further assume that $T_{\mu \nu}$ evaluated on the solution we are considering vanishes, namely
\begin{align}\label{on_shell_T_vanish}
T^*_{\mu\nu}=0,
\end{align}
where ``$*$'' is a reminder that the energy momentum tensor is on-shell. This is the case for (anti)-instantons in euclidean Yang-Mills for example, and it also holds for instanton solutions in the $\Cp {N-1}$ model, as explained earlier.   
The property of the stress-energy tensor being zero on-shell is quite interesting, and we will argue that this property can be used to generate solutions for any $\gamma$, with the aid of some assumptions on $f(T_{\mu\nu})$. Specifically, 
we will assume that $f$ is analytic in $T$, with two further conditions 
\bsubeq\label{a-i-ii}
\bea
f(0)=f(T_{\mu\nu})\Bigg|_{T_{\rho\tau}=T^*_{\rho\tau}=0}=0
,
\\[0.4em]
\frac{\pa f (0)}{\pa T_{\rho\tau}}:=\frac{\pa f (T_{\mu\nu})}{\pa T_{\rho\tau}}\Bigg|_{T_{\rho\tau}=T^*_{\rho\tau}=0}=0
.
\eea
\esubeq
Note that conditions \eqref{a-i-ii} when applied to the case $T^*_{\mu\nu}=0$, if $f$ is analytic, imply that $f$ is at least quadratic in $T_{\mu\nu}$, namely
\bea
f_{{T\overbar{T}}}(T_{\mu\nu})=a\Big(T^{\mu\nu}T_{\mu\nu}-r (T^\mu{}_\mu)^2\Big)+ O(T_{\mu\nu}^3)
.
\label{operator}
\eea
The argument we are going to present solely relies on the previous two assumptions \eqref{a-i-ii}, it holds in any space-time dimensions, and generalises the result of \cite{Ferko:2024yua}, which leveraged the metric approach. Let us denote the fields in the theory by $\varphi (x)$.
If $\varphi^*$ is a solution of the EOM for an undeformed Lagrangian $\cL(\varphi)$ such that  \eqref{a-i-ii} are satisfied for $T^*_{\mu\nu}=T_{\mu\nu}(\varphi^*)$, then $\varphi^*$ is also a solution of the EOM of the deformed Lagrangian  $\cL^{(\gamma)}(\varphi)$ such that $T^{(\gamma)}{}^*_{\mu\nu}=T^*_{\mu\nu}$, \textit{at first order in $\gamma$}. Moreover, the Lagrangian evaluated on the solution is not modified $\cL^{(\gamma)}(\varphi^*)=\cL(\varphi^*)$, again, at least at first order in $\gamma$. The proof is almost immediate, we have
\bea
\cL^{(\gamma)}(\varphi)
=
\cL(\varphi)
+\gamma f_{T\overbar{T}}(T_{\mu\nu})
+\cdots
\eea
the equations of motion of the previous action are schematically
\bea\label{eq:deltaphi}
\frac{\d\cL^{(\gamma)}(\varphi)}{\d\varphi}
=
\frac{\d\cL(\varphi)}{\d\varphi}
+\gamma\frac{\pa f_{T\overbar{T}}(T_{\mu\nu})}{\pa T_{\rho\tau}}\frac{\d T_{\rho\tau}}{\d \varphi}
+\cdots
,
\eea
where we slightly abuse notation and assume that $\delta \varphi$ also includes functional derivatives along partial derivatives of the fields, namely terms like $\delta (\del \varphi)$.
By setting \eqref{eq:deltaphi} to zero, we can schematically write the equations of motion for $\cL^{(\gamma)}$ as
\bea
0
={\rm EOM}[\cL^{(\gamma)},\varphi]
=
{\rm EOM}[\cL,\varphi]
+\l\mathbb{O}_{\rho\tau}\frac{\pa f_{T\overbar{T}}(T_{\mu\nu})}{\pa T_{\rho\tau}}
+\cdots
~.
\eea
Here we denoted with ${\rm EOM}[\cL^{(\gamma)},\varphi]$ and ${\rm EOM}[\cL,\varphi]$ the equations of motion that derive from the variations of the Lagrangians $\cL^{(\gamma)}$ and $\cL$ with respect to the fields $\varphi$, respectively. Moreover, we denoted with
$\mathbb{O}_{\rho\tau}:=\mathbb{O}_{\rho\tau}\left[\frac{\d T_{\rho\tau}}{\d \varphi}\right]$ the operator that arises from the variation of $\frac{\d T_{\rho\tau}}{\d \varphi}$ which could include multiplicative and derivative factors, the latter would arise from integration by parts of terms involving the variation of derivatives of $\varphi$. The explicit form of $\mathbb{O}_{\mu\nu}$ is not important, what matters for us is that it is a linear operator, hence it is such that
$$
\mathbb{O}_{\mu\nu}0=0
.
$$
At this point, it is clear that evaluated on a solution such that
$$
{\rm EOM}[\cL,\varphi^*]=0
,~~~~~~
\frac{\pa f (0)}{\pa T_{\rho\tau}}=\frac{\pa f (T_{\mu\nu})}{\pa T_{\rho\tau}}\Bigg|_{T_{\rho\tau}=T^*_{\rho\tau}=0}=0
,
$$
the variation \eqref{eq:deltaphi} gives
\bea
{\rm EOM}[\cL^{(\gamma)},\varphi^*]
=
{\rm EOM}[\cL,\varphi^*]
+\l\mathbb{O}_{\rho\tau}\frac{\pa f_{T\overbar{T}}(0)}{\pa T_{\rho\tau}}
+\cdots
=0+0+\cdots
,
\eea
hence, $\varphi^*$ is a solution of the equations of motion of the Lagrangian $\cL^{(\gamma)}(\varphi)$, deformed at first order in $\l$. Since $f(0)=0$, we also see that 
$$
\cL^{(\gamma)}(\varphi^*)=\cL(\varphi^*)
,
$$
to first order in $\l$. Moreover, we can compute the Hilbert stress-energy tensor
\bea
  T^{(\gamma)}_{\mu \nu} 
  = 
  - 2 \frac{\partial \cL^{(\gamma)}(\varphi)}{\partial g^{\mu \nu}} 
  + g_{\mu \nu} \cL^{(\gamma)}(\varphi)
  \eea
and going on-shell for $\varphi^*$ we obtain
\bea
    T^*{}^{(\gamma)}_{\mu \nu} 
    &=& 
T^*_{\mu\nu}
    - 2 \frac{\partial f_{T\overbar{T}}(T_{\rho\tau}(\varphi^*))}{\partial g^{\mu \nu}} 
        +\cdots
    \eea
where we used that $\cL^{(\gamma)}(\varphi^*)=\cL(\varphi^*)$ at first order, and that
\bea
  T^*_{\mu \nu} 
  = 
  - 2 \frac{\partial \cL(\varphi)}{\partial g^{\mu \nu}} 
  + g_{\mu \nu} \cL(\varphi)\bigg|_{\varphi=\varphi^*}
  .
  \eea
We now need to evaluate the term proportional to $\frac {\del f_{T\ov T}}{\del g}$ with some care and we will also be more specific on whether we are in flat Minkowski or Euclidean space,
meaning that we will set $g^{\mu\nu}=\eta^{\mu\nu}(\delta ^{\mu\nu})$ at the end of all calculations. This is important since, for now, we have assumed
$\varphi^*$ to be a solution of the Lagrangian in flat space-time, not that we have a lift of the solution on a generic background. We are taking a derivative of $f_{T\bar T}$ with respect to $g$, which can appear either inside of $T_{\mu\nu}$, or contracting products of $T_{\mu\nu}$'s. Since we are assuming $f_{T\ov T}$ to be quadratic in $T$ we must have 
\begin{equation}
    f_{T\ov T}= X^{\mu\nu\rho \sigma}T_{\mu \nu}T_{\rho \sigma}+ O(T^3),
\end{equation}
 where $X$ is some tensor which we take to satisfy $X^{\mu \nu \rho \sigma}=X^{\rho \sigma\mu \nu }$ without loss of generality. We have  
\begin{equation}
    \del _g f_{T\ov T}=  2 X^{\mu\nu\rho \sigma}T_{\mu \nu}\,\del_gT_{\rho \sigma} + O(T^2).
\end{equation}
Crucially we see that $\del _gf_{T\ov T}$ is still linear in $T$, which implies 
\bea
\frac{\partial f_{T\overbar{T}}(T_{\rho\tau}(\varphi^*))}{\partial g^{\mu \nu}} 
&=&
0
\eea
and
\bea
  T^*{}^{(\gamma)}_{\mu \nu} 
  &=& 
T^*_{\mu\nu}
=0.
\eea
This shows that the claim is true at first order. Let us now prove it at second order by using the first-order Lagrangian
\bea
\cL^{(\gamma)}(\varphi)
=
\cL(\varphi)
+\gamma f(T_{\mu\nu})
+\cdots
\eea
as a seed.
From now on, let us denote the undeformed Lagrangian $\cL$ by $\mathcal{L}_0$, which is deformed as 
\begin{align}
    \mathcal{L}_0 \longrightarrow \mathcal{L}^1 = \mathcal{L}_0 + \gamma f \left( T_{\mu \nu}^{(0)} \right) \, , 
\end{align}
where $f \left( T_{\mu \nu}^{(0)} \right)$ is a function of the stress tensor for the seed theory. Given the proof of the first order result,
the idea is that a similar conclusion holds to all orders in $\l$. This is because the leading-order argument can be iterated, since now the deformed theory $\mathcal{L}^1$ can be viewed as a new seed theory, and a similar reasoning shows that a further deformation by a function of the first-order deformed stress tensor $T_{\mu \nu}^{(1)}$ will preserve the solution $\varphi^*$ with $T_{\mu\nu}^*=0$. 

We assume that the Lagrangian has a convergent Taylor series expansion in $\l$,
\begin{align}
    \mathcal{L}^{(\gamma)} = \mathcal{L}_0 + \gamma\mathcal{L}_1 + \l^2 \mathcal{L}_2 + \cdots \, .
\end{align}
We use the symbols $\mathcal{L}_i$ with a lower index for the Taylor coefficients in the Lagrangian, in contrast to the variables $\mathcal{L}^{k}$ with an upper index, which we define as the approximation to $\mathcal{L}^{(\gamma)}$ which is accurate up to $\mathcal{O} ( \l^k )$,
\begin{align}
    \mathcal{L}^{k} = \sum_{i=0}^{k} \l^i \mathcal{L}_i \,~~~
     \mathcal{L}^{(\gamma)} :=\lim_{k\to+\infty} \mathcal{L}^{k}\,.
\end{align}
Likewise, we let $T_{\mu \nu}^{k}$ be the energy-momentum tensor constructed from $\mathcal{L}^k$. By virtue of the differential equation \eqref{appendix_proof_flow}, the approximate Lagrangians $\mathcal{L}^k$ satisfy
\begin{align}\label{taylor_to_stress}
    \mathcal{L}^{k+1} = \mathcal{L}^{k} + \frac{\l^{k+1}}{k+1} \Big[ f \left( T_{\mu \nu}^{k} \right) \Big]_{\l^k} \, ,
\end{align}
where the notation $\left[ f \left( T_{\mu \nu}^{k} \right) \right]_{\l^k}$ means to extract the Taylor coefficient proportional to $\l^k$ in the series expansion of $f \left( T_{\mu \nu}^{k} \right) $. Explicitly,
\begin{align}
    \Big[ g ( \gamma) \Big]_{\l^k} = \frac{1}{k!} \frac{d^k g}{d \l^k} \Big\vert_{\gamma= 0} \, , 
\end{align}
for any function $g(\gamma)$.

Let us look at the second order.
We have
\begin{align}
    \mathcal{L}_1 = f \left( T_{\mu \nu}^0 \right)  , \qquad \mathcal{L}_2 = \frac{1}{2} \Big[ f \left( T_{\mu \nu}^1 \right) \Big]_{\l} , 
\end{align}
where one should also compute derivatives of the function $f$ to extract the Taylor expansion. In particular, it is important that the argument $T_{\mu \nu}^{1}$ of the function $f$ in $\mathcal{L}_2$ is itself determined in terms of the same function $f$:
\begin{align}
     f \left( T_{\mu \nu}^1 \right) &= f \left[ T_{\mu \nu}^{0} + \gamma T_{\mu \nu} \left( \mathcal{L}_1 \right) \right] \nonumber \\
     &= f \left[ T_{\mu \nu}^{0} + \gamma T_{\mu \nu} \left( f \left( T_{\mu \nu}^{(0)} \right) \right) \right]  .
\end{align}
This is because the Hilbert stress tensor is a linear function of the Lagrangian, so in general for a sum $\mathcal{L} = \mathcal{L}_A + \mathcal{L}_B$, the total stress tensor is $T_{\mu \nu} ( \mathcal{L} ) = T_{\mu \nu} ( \mathcal{L}_A ) + T_{\mu \nu} ( \mathcal{L}_B )$.
In particular, assuming appropriate convergence, we also find
\bea
T^{(\gamma)}_{\mu\nu}
=T_{\mu\nu}[\cL^{(\gamma)}]
=\lim_{k\to+\infty}T_{\mu\nu}[\cL^{k}]
=\lim_{k\to+\infty}T^k_{\mu\nu}
.
\eea
The argument above tells us that
$$
{\rm EOM}[\cL_0,\varphi^*]=0
,~~~
T^*_{\mu\nu}=0,
$$
$$
{\rm EOM}[\cL_1,\varphi^*]=0
,~~~
{\rm EOM}[\cL^1,\varphi^*]=0
,~~~
T^1{}^*_{\mu\nu}=0
.
$$
From these, we repeat the derivation of the equations of motion, stress energy-tensor, and so on, to obtain
\bea
{\rm EOM}[\cL^2,\varphi]
=
{\rm EOM}[\cL^1,\varphi^*]
+\hf\l^2\Big{[}\mathbb{O}^1_{\rho\tau}\frac{\pa f(T^1_{\mu\nu})}{\pa T_{\rho\tau}}\Big{]}_{\l},
\label{111111}
\eea
where we have denoted
$$
\mathbb{O}^1_{\rho\tau}=\mathbb{O}^1_{\rho\tau}\left[\frac{\d T^1_{\rho\tau}}{\d \varphi}\right]
$$
which is the operator arising from evaluating $\frac{\d T^1_{\rho\tau}}{\d \varphi}$ in the variation of $\varphi$ followed by potential appropriate integration by parts.
Again, if evaluated on the solution, \eqref{111111} is such that
\bea
{\rm EOM}[\cL^2,\varphi^*]
=
{\rm EOM}[\cL^1,\varphi^*]
+\hf\l^2\Big{[}\mathbb{O}^1_{\rho\tau}\frac{\pa f(0)}{\pa T_{\rho\tau}}\Big{]}_{\l}
=0,
\eea
hence, $\varphi^*$ is a solution of the EOM of the Lagrangian $\cL^2(\varphi)$, deformed at second order in $\l$. Since $f(0)=0$, and given the derivation provided above for $\cL^1$, we have
$$
\cL^{2}(\varphi^*)=\cL^1(\varphi^*)=\cL_0(\varphi^*)
~,~~~\cL_{1}(\varphi^*)=\cL_{2}(\varphi^*)=0.
$$
 From here on, let us assume that $f=f_{T\overbar{T}}$.
The second order stress-energy tensor evaluated on the solutions gives
\bea
  T^{2}_{\mu \nu} 
  = 
  - 2 \frac{\partial \cL^{2}(\varphi)}{\partial g^{\mu \nu}} 
  + g_{\mu \nu} \cL^{2}(\varphi),
  \eea
and evaluated on $\varphi^*$, it is 
\bea
  T^2{}^*_{\mu \nu} 
  &=& 
  - 2 \frac{\partial \cL^{2}(\varphi^*)}{\partial g^{\mu \nu}} 
  + g_{\mu \nu} \cL^{2}(\varphi^*)
\\
  &=& 
  T^1{}^*_{\mu \nu} 
      - \l^2\Bigg\{\Big{[} \frac{\partial f_{T\overbar{T}}(T^1_{\rho\tau}(\varphi))}{\partial g^{\mu \nu}} \Big{]}_{\l}\Bigg\}
      \Bigg|_{g^{\mu\nu}=\eta^{\mu\nu}(\delta^{\mu\nu}),\varphi=\varphi^*}
      =0.
\eea
Here we have used exactly the same calculation done at first order, where it does not matter that we are taking the $[\,]_\l$ truncation. In fact, the argument given above, based on the fact that $f_{T\overbar{T}}$ is 
at least
quadratic in $T_{\mu\nu}$, leads to 
\bea
\frac{\partial f_{T\overbar{T}}(T^1_{\rho\tau}(\varphi))}{\partial g^{\mu \nu}} \Bigg|_{g^{\mu\nu}=\eta^{\mu\nu}(\delta^{\mu\nu}),\varphi=\varphi^*}
=0
~~~\Longrightarrow~~~
\Big{[}\frac{\partial f_{T\overbar{T}}(T^1_{\rho\tau}(\varphi))}{\partial g^{\mu \nu}} \Big{]}_\l\Bigg|_{g^{\mu\nu}=\eta^{\mu\nu}(\delta^{\mu\nu}),\varphi=\varphi^*}
=0.
\eea
By induction, assuming the argument to be true up to order $k$,
$$
{\rm EOM}[\cL^k,\varphi^*]=0
,~~~
T^k{}^*_{\mu\nu}=0.
$$
Upon repeating the same argument, we find 
\bea
{\rm EOM}[\cL^{k+1},\varphi]
=
{\rm EOM}[\cL^k,\varphi]
+ \frac{\l^{k+1}}{k+1} \Big{[}\mathbb{O}^k_{\rho\tau}\frac{\pa f(T^k_{\mu\nu})}{\pa T_{\rho\tau}}\Big{]}_{\l^k}
,
\eea
where we have denoted
$$
\mathbb{O}^k_{\rho\tau}=\mathbb{O}^1_{\rho\tau}\left[\frac{\d T^k_{\rho\tau}}{\d \varphi}\right],
$$
which is the linear operator arising from evaluating $\frac{\d T^k_{\rho\tau}}{\d \varphi}$ in the variation of $\varphi$ followed by potential appropriate integration by parts.
Again, if evaluated on the solution, we have
\bea
{\rm EOM}[\cL^{k+1},\varphi^*]
=
{\rm EOM}[\cL^k,\varphi^*]
+ \frac{\l^{k+1}}{k+1} \Big{[}\mathbb{O}^k_{\rho\tau}\frac{\pa f(0)}{\pa T_{\rho\tau}}\Big{]}_{\l^k}
=0
.
\eea
Hence, $\varphi^*$ is a solution of the equations of motion associated to $\cL^{k+1}(\varphi)$, deformed at order $k+1$ in $\l$. Given that $f(0)=0$, and the derivation given above for $\cL^k$, we also see that 
\bea
\cL^{k+1}(\varphi^*)=\cL_0(\varphi^*)
.
\eea
The $(k+1)$-th order stress-energy tensor 
\bea
  T^{k+1}_{\mu \nu} 
  = 
  - 2 \frac{\partial \cL^{k+1}(\varphi)}{\partial g^{\mu \nu}} 
  + g_{\mu \nu} \cL^{k+1}(\varphi)
  \eea
evaluated on the solution $\varphi^*$, is 
\bea
  T^{k+1}{}^*_{\mu \nu} 
  &=& 
  - 2 \frac{\partial \cL^{k+1}(\varphi^*)}{\partial g^{\mu \nu}} 
  + g_{\mu \nu} \cL^{k+1}(\varphi^*)
\\
  &=& 
  T^k{}^*_{\mu \nu} 
      - 2 + \frac{\l^{k+1}}{k+1} \Bigg\{\Big{[}\frac{\partial f_{T\overbar{T}}(T^1_{\rho\tau}(\varphi))}{\partial g^{\mu \nu}} \Big{]}_{\l^k}\Bigg\}\Bigg|_{g^{\mu\nu}=\eta^{\mu\nu}(\delta^{\mu\nu}),\varphi=\varphi^*}
      =0
.
\eea
Here we used again the same argument as above to obtain
\bea
\frac{\partial f_{T\overbar{T}}(T^k_{\rho\tau}(\varphi))}{\partial g^{\mu \nu}} \Bigg|_{g^{\mu\nu}=\eta^{\mu\nu}(\delta^{\mu\nu}),\varphi=\varphi^*}
=0
~~~\Longrightarrow~~~
\Big{[}\frac{\partial f_{T\overbar{T}}(T^k_{\rho\tau}(\varphi))}{\partial g^{\mu \nu}} \Big{]}_{\l^k}\Bigg|_{g^{\mu\nu}=\eta^{\mu\nu}(\delta^{\mu\nu}),\varphi=\varphi^*}
=0
.
\eea
Hence, the induction argument gives
$$
\forall k~,~~~
{\rm EOM}[\cL^k,\varphi^*]=0
~,~~~
T^k{}^*_{\mu\nu}=0,
$$
and taking the $k\to+\infty$ limit while assuming convergence implies
\bea
\cL^{(\gamma)}(\varphi^*)=\cL_0(\varphi^*)
,~~~
{\rm EOM}[\cL^{(\gamma)},\varphi^*]=0
,~~~
T^{(\gamma)}{}^*_{\mu\nu}=0
.
\eea
This indicates that a solution of the equations of motion, $\varphi^*$ with $T^*_{\mu\nu}=0$, is also a solution to the full equations of motion derived from the Lagrangian $\cL^{(\gamma)}$, and that $T^{(\gamma)}_{\mu\nu}(\varphi^*)$ remains zero along the whole flow together with the relation $\cL^{(\gamma)}(\varphi^*)=\cL_0(\varphi^*)$ for the on-shell Lagrangian. This argument applies, in particular, to the instanton solutions of the $\Cp {N-1}$ model described in section \ref{sec:CPN_review}.

In order to extend beyond the case of $\TTbar$-like flows, and further study the effect of higher spin deformations, we will now look at their auxiliary field formulation. This is particulary powerful when we look at soliton surfaces of the deformed $\Cp 1$ model. We will also see that the non-deformation of instantons by $T\ov T$-like flows is manifest in the auxiliary field formulation and we will propose a generalisation of the ``vanishing $T$" condition appropriate for more general higher spin deformations.

\section{Auxiliary field formalism}\label{s:cpnaux}
In this section we focus on using auxiliary fields to describe higher spin deformations of the $\Cp {N-1}$ model. We will begin with a brief review of \cite{Bielli:2024ach,Bielli:2024oif}, and then specialise to $\Cp{N-1}$.
 
 \subsection{Review of (higher-spin) auxiliary field deformations}\label{ss:auxrev}

We have mentioned, in Section \ref{s:ttbcpn}, that the $\TT$ deformation preserves classical integrability. Deformations by specific higher-spin analogues of $\TT$, the so-called Smirnov-Zamolodchikov operators \cite{Smirnov:2016lqw}, also preserve integrability. Surprisingly, it seems that more is true, at least for seed theories which belong to certain classes of $2d$ integrable sigma models: deformations by \emph{arbitrary functions} of the stress tensor and high-spin conserved quantities preserve classical integrability. An elegant and powerful way to package the entire family of all such deformed models is the auxiliary field formalism, introduced in \cite{Ferko:2024ali} and building on \cite{Borsato:2022tmu,Ferko:2023wyi}.

The main idea of the auxiliary field machinery is to implement integrable deformations of a $2d$ integrable field theory by introducing a vector field $v$, which appears algebraically in the action and couples to specific model-dependent currents. Any coupling constants, such as the $\TT$ parameter $\gamma$, associated to the integrable deformation \textit{only} appear in the self interaction terms of $v$, through an ``interaction function'' $E$. This formalism is surprisingly general. It can be applied systematically to deform a large class of theories such as the principal chiral model (with \cite{Fukushima:2024nxm,Bielli:2024ach} or without \cite{Ferko:2024ali} Wess-Zumino term), the Yang-Baxter and bi-Yang-Baxter deformed PCM \cite{Bielli:2024fnp}, and various coset sigma models \cite{Cesaro:2024ipq,Bielli:2024oif}. These auxiliary field couplings have also been engineered in the framework of $4d$ Chern-Simons theory \cite{Fukushima:2024nxm}, and find applications to dimensionally reduced gravity \cite{Cesaro:2025msv}.

The only class of deformations that will be relevant for us are auxiliary field symmetric space sigma models; hence, we will omit the other cases from the review.  In this case, the fundamental object in the seed theory is a coset-valued field 
\begin{equation}
    g:\Sigma \longrightarrow G/H, 
\end{equation}
where $G$ is a Lie group and $H$ is a subgroup of $G$, $H < G$.

Note that we will tacitly assume $g$ to be composed with a local section of the fibration 
\begin{equation}
    G\overset \pi \longrightarrow G/H,
\end{equation}
in other words, $g$ is a coset representative of $G/H$ pulled back to the worldsheet. Accordingly, one can decompose the Lie algebra $\fg$ of $G$ as 
\begin{equation}\label{lie_algebra_decomp}
    \fg = \fg ^{(0)}\oplus \fg ^{(2)},
\end{equation}
where $\fg ^{(0)}$ is the algebra of $H$ and $\fg ^{(2)}$ its orthogonal complement with respect to the $G$-invariant  bilinear form $\ttr{-,-}$, namely 
\begin{equation}
    \ttr {\fg^{(2)}\fg ^{(0)}}=0. 
\end{equation}
We also assume the following properties:
\begin{equation}\label{eq:cosetrels}
    \begin{split}
        [\fg ^{(2)},\fg ^{(2)}] &\subset \fg ^{(0)}, \quad [\fg ^{(0)},\fg ^{(2)}]\subset \fg ^{(2)}. 
    \end{split}
\end{equation}
One can define the $\fg$-valued worldsheet current 
\begin{equation}
    j = g^{-1}dg,
\end{equation}
which satisfies the Maurer-Cartan identity 
\begin{equation}\label{eq:cosetMC}
    \del _\mu j_\nu - \del _\nu j_\mu+[j_\mu,j_\nu]=0.
\end{equation}
In turn, it is useful to project $j$ onto its $(0)$ and $(2)$ components and write 
\begin{equation}
    j= j^{(0)}+j^{(2)}.
\end{equation}
Decomposing \eqref{eq:cosetMC} into its $\fg ^{(0)}$ and $\fg ^{(2)}$ blocks leads to two independent conditions 
\begin{equation}\label{mc_projection_first_time}
    F^{(0)}_{\mu\nu}+[j^{(2)}_\mu,j^{(2)}_\nu]=0,\quad {\cal D}_\mu j^{(2)}_\nu-{\cal D}_\nu j^{(2)}_\mu=0,
\end{equation}
where 
\begin{equation}
    F^{(0)}_{\mu\nu}=     \del _\mu j^{(0)}_\nu - \del _\nu j^{(0)}_\mu+[j^{(0)}_\mu,j^{(0)}_\nu],
\end{equation}
and we have introduced the $H$-covariant derivative 
\begin{equation}
    {\cal D}_\mu = \del _\mu + [j_\mu ^{(0)}, - ].
\end{equation}
The action of the undeformed coset model is then simply 
\begin{equation}\label{undeformed_coset_action}
    S[g]=\int_\Sigma\frac 12 \ttr {j_\mu^{(2)}j_\nu^{(2)}}\eta ^{\mu\nu}\, d^2x,
\end{equation}
where $\eta$ is the $2d$ Minkowski metric. The equations of motion can be conveniently expressed as a conservation equation in terms of $j$ and ${\cal D}$
\begin{equation}\label{undeformed_coset_eom}
    {\cal D}^\mu j_\mu ^{(2)}=0. 
\end{equation}
From now on, we will use lightcone coordinates in Lorentzian signature, such that 
\begin{equation}
    x ^\pm=\frac 12 (x^1 \pm x^2 ), 
\end{equation}
note that these conventions are different from the complex Euclidean version \eqref{complexifiedcoordinate}.
The Minkowski metric becomes 
\begin{equation}
    \eta_{\mu\nu} =\begin{pmatrix}
       0 &-2\\
        -2&0
    \end{pmatrix}. 
\end{equation}
In this coordinate system the action reads  
\begin{equation}\label{eq:lightacn}
    S[g]=\int_\Sigma-\frac 12 \ttr {j_+^{(2)}j_-^{(2)}}\, d^2x.
\end{equation}

In order to describe higher spin deformations we will couple \eqref{eq:lightacn} to a vector auxiliary field $v^{(2)}\in \Omega ^1(\Sigma )\otimes \fg^{(2)}$. The deformed action is then given by 
\begin{equation}\label{eq:auxacn}
    \begin{split}
        S[g,v] = \int_\Sigma \Bigg(&-\frac 12 \ttr {j_+^{(2)}j_-^{(2)}}+\ttr {v_+^{(2)}j_-^{(2)}}+\ttr {j_+^{(2)}v_-^{(2)}}\\[0.4em]
        &+ \ttr {v_+^{(2)}v_-^{(2)}}+E\bigg(\ttr {(v_+^{(2)})^k},\ttr {(v_-^{(2)})^l}\bigg)\Bigg) d^2x,
    \end{split}
\end{equation}
where $E$ is the interaction function. We will sometimes refer to (\ref{eq:auxacn}) as the AF-SSSM, for ``auxiliary field symmetric space sigma model''. Note that it is crucial for $E$ to be Lorentz invariant and to only be a function of ``chiral traces''. Namely, a sufficient condition to preserve integrability is that only traces of products of $v^{(2)}_{\pm}$ with homogeneous signs appear in $E$, not objects like $\ttr {(v_+^{(2)})^p\,(v_-^{(2)})^q}$ where both $p$ and $q$ are non-vanishing \cite{Bielli:2024oif,Bielli:2025uiv}.

The equations of motion can be derived by varying \eqref{eq:auxacn} with respect to $g$ and $v$, giving
\begin{equation}\label{eq:coseteom}
    \begin{split}
        v^{(2)}_\pm +j^{(2)}_\pm +\Delta_{\pm}&=0,\\[0.4em]
        {\cal D}_+\mathfrak J^{(2)} _- + {\cal D}_-\mathfrak J^{(2)}_+&=2\big( [v_-^{(2)},j^{(2)}_+]+[v_+^{(2)},j^{(2)}_-] \big),
    \end{split}
\end{equation}
where
\begin{align}\label{frakJ_defn}
    \mathfrak J_\pm = -(j_\pm + 2 v _\pm )
\end{align}
and $\Delta _\pm = \delta _{v^{(2)}_\mp}E$ comes from the variation of the interaction function along $v$. Note that since $E$ only depends on chiral traces, we find that $\Delta_\pm $ commutes with $v^{(2)}_\mp$. In turn, this implies that if we are on shell for the auxiliary fields, namely we treat their equations of motion as a constraint, then the current $\mathfrak J$ becomes conserved
\begin{equation}\label{eq:onshelleom}
    {\cal D}_+\mathfrak J^{(2)} _- + {\cal D}_-\mathfrak J^{(2)}_+=0.
\end{equation}
Crucially, when the auxiliary field equations of motion are satisfied, one has the identities
\begin{align}\label{commutator_identities}
    [ \mathfrak{J}_+^{(2)} , \mathfrak{J}_-^{(2)} ] = [ j_+^{(2)} , j_-^{(2)} ] \, , \qquad [ \mathfrak{J}_+^{(2)} , j_-^{(2)} ] = [ j_+^{(2)} , \mathfrak{J}_-^{(2)} ] \, ,
\end{align}
which can be used \cite{Bielli:2024oif} to show\footnote{In comparing to \cite{Bielli:2024oif}, note that our conventions for the Lax here differ by an overall sign as discussed around (\ref{bad_sign_lax}).} that \eqref{eq:onshelleom} is equivalent to flatness of the Lax connection
\begin{equation}\label{eq:defLax}
     L_\pm=- j^{(0)}_{\pm}- \frac {(z^2+1)j_\pm ^{(2)} \mp 2z\mathfrak J^{(2)} _\pm}{z^2-1} \, .
\end{equation}
Therefore any deformation of this type is weakly integrable. It is important to stress that classical integrability is a consequence of the interaction function \textit{only} depending on chiral traces. If this were not the case, $\Delta_\pm$ would fail to commute with $v_\mp$ and $\mathfrak J^{(2)}$ would not necessarily be conserved. It could be that integrable deformations where $E$ contains non-chiral traces exist, however, the Lax associated to them is not \eqref{eq:defLax}. 

This brief review was not intended to be complete, as we have limited ourselves to stating the aspects of this formalism which will be relevant for us. For more details on the standard SSSM, the interested reader should consult Section 4 of \cite{Zarembo:2017muf}, Section 2.2.3 of \cite{Seibold:2020ouf}, or Section 1.3 of \cite{yoshida2021yang}; for further information on auxiliary field deformations, see the papers \cite{Ferko:2024ali,Bielli:2024ach,Bielli:2024oif}. We will now discuss how to apply this technology to the $\Cp {N-1}$ model.
\subsection{Soliton surfaces of $\Cp1$ using auxiliary fields}
After reviewing the soliton surfaces of the $\Cp 1$ model in section \ref{STCPN}, it is interesting to understand how these are affected by integrable higher spin deformations. The simplest way to probe the geometry is again to look at the Gauss and mean curvatures, as we did for the undeformed theory. 

 As anticipated, the auxiliary field formalism is especially well suited to study these types of integrable deformations. Crucially, the $v$ field couples to specific Lie algebra-valued currents in the undeformed theory. Hence, in order to exploit this technology, we must first re-write the $\Cp {N-1}$ model as a coset theory on 
\begin{equation}
    \Cp{N-1}\cong \frac {U (N)}{U (N-1)\times U (1)}\cong \frac{SU (N)}{SU(N-1)\times U (1)}.  
\end{equation}
\subsubsection{The $\Cp{N-1}$ model as a coset}\label{sec:CPN_as_coset}
The main ingredient we need is a coset representative written in terms of the standard set of fields we have been using, namely $\phi_I$, with $I=1,\ldots ,N$. Equivalently, a local section 
\begin{equation}
    A: \Cp {N-1}\rightarrow U (N).
\end{equation}
For simplicity, we will write $A$ in the fundamental representation. It is convenient to split it into blocks according to the decomposition 
\begin{equation}
    \mathbf N\rightarrow (\mathbf{N-1})\oplus \mathbf 1,
\end{equation}
which arises from breaking $U (N)\rightarrow U(N-1)\times U( 1)$. Let us decompose the $I$ index as $I=(1,i)$, where $1$ labels the $U(N-1)$ singlet direction, and $i = 2,\ldots ,N$, the fundamental. With these conventions, we can choose 
\begin{equation}
    \begin{split}
        A^1{}_1 &=\phi_1,\\
        A^i{}_1 &= \phi _i,\\
        A^1{}_i &= \frac {\phi _1}{|\phi_1|}\, \ov \phi _i,\\
        A^i{}_j&=\frac 1{1+|\phi_1|}\, \phi _i\, \ov \phi _j - \delta _{ij}, 
    \end{split}
\end{equation}
where $\ov \phi$ is the complex conjugate of $\phi$. \\

We will follow the notational conventions of \cite{Bielli:2024oif}, in this language, the $\uu N$-valued current is given by 
\begin{equation}
    j_{IJ}=\ov A^K{}_I \,dA^K{}_J.
\end{equation}
In order to recover \eqref{CPNLagrangian1}, we only need the components of $j$ in the coset algebra 
\begin{equation}
    \mathfrak g ^{(2)}= \uu N\ominus (\uu {N-1}\oplus \uu 1), 
\end{equation}
namely 
\begin{equation}\label{eq:j2dec}
    j^{(2)}= \big( j_{1i},\, j_{i1},\, 0\big). 
\end{equation}
It is then straightforward to check that
\begin{equation}\label{eq:CPcoset}
    \begin{split}
        \cL = \frac12 \ttr {j^{(2)}_\mu\, j^{(2)}_\nu}\eta ^{\mu\nu}= - d\phi ^I\, d \ov \phi ^I -  \left ( \phi ^I \, d \ov \phi ^I\right )^2,
    \end{split}
\end{equation}
which is consistent with \eqref{CPNLagrangian1} up to an overall rescaling of the Cartan-Killing form. Note also that the expression (\ref{eq:CPcoset}) for the Lagrangian, $\mathcal{L} \sim \tr ( j_+^{(2)} j_-^{(2)} )$, takes the same form $\mathcal{L} \sim \tr ( K_+ K_- )$ as in the projector formalism described around equation (\ref{LagrangianinK}); the objects $j_\pm^{(2)}$ and $K_\pm$ have some similar properties, but are different ways of parameterizing the degrees of freedom in the theory.

Higher spin integrable deformations of \eqref{eq:CPcoset} can now be constructed using the AFSM formalism of section \ref{ss:auxrev}.
\subsubsection{Deformed soliton surfaces of $\Cp 1$}\label{s:defsolaux}
We will now look at how the soliton surface curvatures associated to the $\Cp 1 \cong S^2$ model are affected by higher spin deformations. As anticipated, we will see that the auxiliary field formalism greatly simplifies these calculations. Before moving on to the soliton surfaces, let us take a step back and analyse what kind of deformations are allowed in the $\Cp {N-1}$ model. According to the decomposition \eqref{eq:j2dec}, the field $v_\pm^{(2)}$ must take the form  
\begin{equation}
    {v^{(2)}_\pm}^I{}_J= \begin{pmatrix}
        0_{1\times 1}&w_{\pm,i}\\
        y_\pm ^i&0_{(N-1)\times (N-1)}
    \end{pmatrix},
\end{equation}
where $w$ and $y$ are $N-1$ covectors and vectors respectively. Note that we are working in the complexified Lie algebra, so we do not necessarily assume $\ov w = - y$. Squaring $v^{(2)}$ gives 
\begin{equation}
    (v_\pm^{(2)})^2=\begin{pmatrix}
         w_\pm\cdot y_\pm&0\\
        0&y_\pm^i \, w_{\pm,j}
    \end{pmatrix} := X_{\pm\pm},
\end{equation}
and similarly 
\begin{equation}
    \begin{split}
        v_\pm^{(2)}\, X_{\pm\pm} &= ( w_\pm \cdot y_\pm ) \, v_\pm^{(2)},\\[0.4em]
        (X_{\pm\pm})^2 &=( w_\pm \cdot y_\pm )\, X_{\pm\pm}.
    \end{split}
\end{equation}
It is then clear that chiral traces simplify to 
\begin{equation}
    \begin{split}
        \ttr {(v_\pm^{(2)})^{2k-1}}=0, \quad \ttr {(v_\pm^{(2)})^{2k}}=\frac 1{2^{k-1}}\ttr {(v_\pm^{(2)})^{2}}^k.
    \end{split}
    \label{tracescpn}
\end{equation}
The key takeaway from this simple calculation is that it is sufficient to consider higher spin deformations of the form 
\begin{equation}\label{eq:cpE}
    E=E(\nu_2), \quad \nu _2=\ttr {(v^{(2)}_+)^2}\ttr {(v^{(2)}_-)^2}.
\end{equation}
With this choice the constraint arising from the auxiliary fields simplifies to 
\begin{equation}\label{eq:veomcp1}
    j^{(2)}_\pm+v^{(2)}_\pm + 2 E' \ttr {(v^{(2)}_\pm)^2}v^{(2)}_\mp=0.
\end{equation}
We will compute the curvatures on soliton surfaces only in terms of $v^{(2)}$. We can eliminate $j^{(2)}$ from the Lax connection \eqref{eq:defLax} using \eqref{eq:veomcp1}
\begin{equation}\label{eq:defLaxv}
     L _\pm =-\jz \pm + \frac {z\mp 1 }{z\pm 1}\, \vt \pm + \frac {z\pm 1}{z\mp 1 }\, 2 E' \, \ttr {(\vt \pm )^2 }\, \vt \mp.
\end{equation}
Before moving on to computing quantities on the surface, let us make an observation on the structure of $L$. Firstly, it will be useful to decompose it into $(0)$ and $(2)$ blocks
\begin{equation}
     L= L^{(0)}+L^{(2)}.
\end{equation}
The ``undeformed Lax'' is
\begin{equation}\label{eq:undefLx}
     L_\pm\big|_{E=0}=-\jz \pm + \frac {z\mp 1 }{z\pm 1}\, \vt\pm .
\end{equation}
Note that calling \eqref{eq:undefLx} undeformed is incorrect, because it is evaluated on solutions of the deformed theory. Namely, it would not be flat if $v^{(2)}$ and $\jz{}$ corresponded to solutions of the theory with some non-trivial $E$. However, it formally corresponds to \eqref{eq:defLax} upon setting $E=0$. It is convenient to introduce this object, because it satisfies 
\begin{equation}
      L^{(2)} _\mu = \Lambda _\mu {}^\nu \,  L^{(2)}_\nu \big|_{E=0},
\end{equation}
where 
\begin{equation}\label{other_lax_lambda}
    \Lambda=\begin{pmatrix}
        1 & 2E'\ttr {(\vt +)^2}\\
        2E'\ttr {(\vt -)^2}&1
    \end{pmatrix},
\end{equation}
where we are still assuming that the deformation parameter $\gamma $ is small enough that $\Lambda $ is not orientation reversing.
In short, this is saying that the Lax connection where we formally set $E=0$, is related to \eqref{eq:defLax} by a local linear transformation. 

We can now use this fact to compute the curvatures. Substituting \eqref{eq:defLaxv} into \eqref{pullback_metric} immediately yields the deformed metric $g_{\mu \nu}$ (not to be confused with the $G$-valued field $g$; whether we are talking about one or the other will always be obvious from context) on a soliton surface 
\begin{equation}\label{deformed_soliton_metric}
    \begin{split}
        g_{++} &=-4 \,\ttr {\vt + \vt + } \left ( \frac 1 {(z+1)^4}- \frac {4 E' }{(z^2-1)^2}\, \ttr {\vt + \vt - } + \frac 4 {(z-1)^4}\, (E')^2 \, \nu_2 \right ),\\[0.6em]
          g_{--}&=-4 \,\ttr {\vt - \vt - } \left ( \frac 1 {(z-1)^4}- \frac {4 E' }{(z^2-1)^2}\, \ttr {\vt + \vt - } + \frac 4 {(z+1)^4}\, (E')^2 \, \nu_2 \right ),\\[0.6em]
          g_{+-}&=\frac 4{(z^2-1)^2}\ttr {\vt + \vt -} \left (1+4\, (E')^2\,\nu_2\right )- 8\, E'\, \nu_2 \left ( \frac 1 {(z+1)^4}+ \frac  1 {(z-1)^4}\right ).
    \end{split}
\end{equation}
Notice that because \eqref{pullback_metric} is tensorial in the Lax, and $\Lambda$ does not depend on the spectral parameter $z$, the deformed metric formally corresponds to 
\begin{equation}\label{metric_Lambda_relation}
    g= \Lambda \big(g\big |_{E=0}\big) \Lambda ^T.
\end{equation}
The second fundamental form is in principle given by 
\begin{equation}\label{eq:secf}
    h_{\mu \nu} = \big |[\del _zL _+, \del _z L_- ]\big |^{-1}\, \langle  \del _z \, \del _\mu L _\nu + \left [ \del _z \, L _\nu , L _\mu \right ], \, [\del _zL _+, \del _z L_- ]\rangle . 
\end{equation}
However, using the fact that we are in a coset model simplifies this expression even if we stay formally off-shell. More precisely, we can use \eqref{eq:cosetrels} to eliminate several contributions to $h_{\mu\nu}$, and end up with the only non-vanishing term being 
\begin{equation}\label{eq:dsimp}
    h_{\mu\nu} = \big |[\del _zL _+, \del _z L_- ]\big |^{-1}\, \langle  \left [ \del _z \, L _\nu , L^{(2)} _\mu \right ], \, [\del _zL _+, \del _z L_- ]\rangle . 
\end{equation}
It is now manifest that $h_{\mu\nu}$ is itself tensorial in $L$, hence we again have 
\begin{equation}
    h= \Lambda \big(h\big |_{E=0}\big) \Lambda ^T.
\end{equation}
Recall that the mean and Gauss curvatures $H$ and $K$ can be written in terms of the metric and second fundamental form as 
\begin{equation}\label{eq:kappaom}
    \begin{split}
            H &= \frac 12 \, \ttr { h\, g^{-1} },\\
            K&=   \det \left (h\, g^{-1}\right ).
    \end{split} 
\end{equation}
We can formally compute these in the undeformed case $E=0$, and because all $\Lambda$ contributions will be cancelled in \eqref{eq:kappaom}, the form of the result in terms of $\vt {}$ will stay unchanged. We find 
\begin{equation}\label{eq:curvdef}
    \begin{split}
        H &= (z^2-1)\, \frac {\langle \vt +, \vt -\rangle}{\big|[\vt +, \vt -]\big |},\\
        K &= -\frac 12 \, (z^2 -1 )^2,
    \end{split}
\end{equation}
which is consistent with the undeformed expression in Eq. \eqref{Curvoriginal}. Note that $H$ in Eq. \eqref{eq:curvdef} is not undeformed because it is evaluated on a  \textit{deformed} solution. However, the constant Gauss curvature $K$ remains undeformed. 

\subsection{Non-deformation of instanton solutions for analytic auxiliary field deformations}

We have seen how powerful the auxiliary field formalism is when looking at the soliton surfaces of higher-spin deformed models. It is then natural to ask whether the conservation of instanton solutions of section \ref{s:instpres1} also fits naturally within this description. We will see now that the answer is yes, and, in particular, the argument via auxiliary fields turns out to be almost immediate and barely requires any calculations. In addition, we will propose a generalisation of the ``vanishing $T$'' condition for more general higher-spin deformations.

We start by expressing $\ttr {(\jt {})^2}$ in terms of the fields $\phi$ of the $\Cp {N-1}$ model, 
\begin{equation}
    \begin{split}
        \ttr {j^{(2)}_+\,j^{(2)}_+ }&= D_{+}\,\phi ^I\,\ov {(D_{-}\, \phi ^I)},\\[0.4em]
        \ttr {j^{(2)}_{-}\,j^{(2)}_{-} }&=D_{-}\,\phi ^I\,\ov {(D_{ +}\, \phi ^I)}, 
    \end{split}
\end{equation}
where $D$ is defined in \eqref{defcovariantderivative}, and we are now in Euclidean signature, so $\pm$ indices label complex coordinates. The instanton condition \eqref{topologicalchiralequation} then implies 
\begin{equation}\label{eq:instantonjj}
    \ttr {j^{(2)}_+\,j^{(2)}_+ }=\ttr {j^{(2)}_{-}\,j^{(2)}_{-} }=0,
\end{equation}
on instanton solutions.
It is now straightforward to argue that these solutions are undeformed by analytic higher-spin deformations. 
\begin{enumerate}
    \item It is clear that setting 
    \begin{equation}\label{eq:vjrel}
        \jt{}=\jt{\rm instanton}, \quad \vt {}=-\jt {} ,
    \end{equation}
    solves the $g$ equation of motion in \eqref{eq:coseteom}, namely
    \begin{equation}
         {\cal D}_+\mathfrak J^{(2)} _- + {\cal D}_-\mathfrak J^{(2)}_+=2\big( [v_-^{(2)},j^{(2)}_+]+[v_+^{(2)},j^{(2)}_-] \big).
    \end{equation}
    We see that the right hand side simply vanishes by anti-symmetry of the Lie bracket, and $\mathfrak J$ reduces to $\jt{}_{\rm instanton}$. Hence, the full expression reduces to the equation of motion of the undeformed theory, which is satisfied by assumption.\label{bp:1}
    \item Recall the constraints imposed by the $\vt {}$ equations of motion in \eqref{eq:coseteom} 
    \begin{equation}
        v^{(2)}_\pm +j^{(2)}_\pm =-\Delta_{\pm}.
    \end{equation}
    The right hand side vanishes by \eqref{eq:instantonjj} and the fact that $\Delta$ becomes an analytic function of traces of $\jt{}$, while the left vanishes automatically by \eqref{eq:vjrel}.\label{bp:2}
\end{enumerate}

In conclusion, we see that any instanton solution will automatically lift to the deformed theory, provided the deformation is analytic in traces. In particular, this is the case for all those theories that arise as a Smirnov-Zamolodchikov flow. 

By Eq. \eqref{eq:cpE}, we are essentially in the same situation as in section \ref{s:instpres1}. Perhaps more interestingly, we see that points \ref{bp:1} and \ref{bp:2} above revolve around a single relation:
\begin{equation}\label{eq:Dis0}
    \Delta \big|_{v=j_{\rm instanton}}=0.
\end{equation}
So we see that any solution of the theory satisfying \eqref{eq:Dis0} is by construction preserved by the deformation. Crucially, this argument will hold for any auxiliary field sigma model, not exclusively for the symmetric space case. We can then write the following general statement:\\

\textit{Consider an auxiliary field sigma model with interaction function $E$ which depends on a set of Lorentz invariant variables $\{x^i\}_{i=1\ldots n}$. Any solution $\widetilde j$ of the seed theory such that 
\begin{equation}\label{eq:criteria1}
    \nabla E\Big|_{v=-\widehat j}=0,
\end{equation}
where $\nabla E$ is the gradient of $E$, lifts to a solution of the AFSM.} 
\\

It would be particularly interesting to attempt a classification of the moduli space of solutions that satisfy \eqref{eq:criteria1}, at least for some class of AFSMs.
  
\subsubsection{Instanton solutions for Root-$T\overbar{T}$ case}
In the previous discussion, as well as in section \ref{s:instpres1}, the key property that solutions had to satisfy in order to be preserved by higher spin deformations was 
\begin{equation}
    \Delta_{\text{on-shell}}=0.
\end{equation}
To ensure this, we always assumed deformations to be ``analytic'' in an appropriate sense. However, it is often interesting to look at deformations that are not analytic in any meaningful way. One example is the root-$T\ov T$ deformation \cite{Ferko:2022cix}, which is both non-analytic 
and (classically) conformal. As such, it is an open question to determine whether it is exactly marginal or marginally irrelevant. 
The deformation is associated with the flow equation
\begin{equation}\label{rtt_defn}
    \del _\gamma \cL^{(\gamma)}
    =
    \sqrt{ \frac{1}{2}T^{(\gamma)}{}^{\mu \nu} T^{(\gamma)}_{\mu \nu} 
    - \frac{1}{4} \left( {T}^{(\gamma)}{}{^\mu{}_\mu} \right)^2 } 
    =\sqrt{-\det{\left[t^{(\gamma)}{}^{\mu\nu}\right]}}
    \,,
\end{equation}
with $t^{\mu\nu}=T^{\mu\nu}-\frac{1}{2}\eta^{\mu\nu}T^{\rho}{}_{\rho}$ the traceless part of the energy-momentum tensor.
In the case of (classical) conformal field theories, the deforming operator becomes the square root of the $T\ov T$ operator, hence justifying the $\sqrt{T\ov T}$ nomenclature.

In auxiliary field language, $\sqrt {T\ov T}$ is described by setting 
\begin{equation}\label{rtt_interaction}
    E= \tanh \left(\frac \gamma 2\right)\,\sqrt {\nu _2}. 
\end{equation}
The integrated Lagrangian was already known before the advent of auxiliary fields methods and can be written in terms of $\jt {}$ \cite{Borsato:2022tmu} (see Appendix \ref{ap:rootlag} for the derivation starting from the AFSM and integrating out $\vt {}$):
\begin{equation}\label{eq:Lj}
\cL = -\frac 12 \left( \cosh\left(\frac \gamma 2\right)\, \tr (\jt+\jt-)+ \sinh \left(\frac \gamma 2 \right)\,
     \sqrt {\tr ((\jt+)^2)\tr ((\jt-)^2)}\right). 
\end{equation}
 Since we determined that instantons are preserved along analytic SZ flows, it is natural to ask whether the same is true for non-analytic examples such as $\sqrt{T\ov T}$.  However, in this case there is a simple obstruction to the argument we presented in the previous section, namely 
\begin{equation}\label{eq:deltroot}
    \Delta_\pm= \tanh \left(\frac\gamma 2\right)\frac{\ttr{(\vt\pm)^2}}{2\sqrt {\nu_2}}\, \vt \mp, 
\end{equation}
hence, any field configuration with $\ttr{(\vt\pm)^2}=0$ has a divergent $\Delta$ and appears not to be a solution of the equations of motion. \\

We will now give a brief argument to show that despite this pathology, instanton solutions still minimize the action. 
Note that we must go to Euclidean signature in order to consider instantons. Accordingly, we should flip the sign in front of the kinetic term $\tr(\jt +\jt -)$ in \eqref{eq:Lj}. The extremisation of the action follows in a few simple steps:
\begin{itemize}
    \item It seems naively impossible to set $\vt {}=-\jt{}$ and $\jt {}= \jt{\rm instanton}$ because equation \eqref{eq:deltroot} blows up. 
    \item We have already shown in \eqref{inequality4} that $\jt{\rm instanton} $ minimise the kinetic term 
    \begin{equation}
        \tr(\jt+\jt-)
    \end{equation}
    in
    \eqref{eq:Lj}. 
    \item We have also seen that the second term in \eqref{eq:Lj}, vanishes on shell for the instantons, beecause it is proportional to $\sqrt {\tr ((\jt+)^2)\tr ((\jt-)^2)}$.
    \item Assuming the positive branch for the square root, it becomes obvious that $\jt {\rm instanton}$ also minimises the whole action. 
\end{itemize}
It would be interesting to understand these deformations when the negative branch is chosen instead, and hence it is not clear that the action is being extremised.

\section{Geometric interpretation}\label{sec:geometric}
The preceding sections of this work have investigated integrable deformations of the $\mathbb{CP}^{N-1}$ model -- including the $\TT$ and root-$\TT$ flows, and more general auxiliary field deformations -- focusing on aspects including soliton surfaces and instanton solutions. Because these solvable deformations appear to preserve certain desirable features of the seed model in a quite universal way, it is natural to wonder whether there is a unifying perspective on this class of models which make some of their properties manifest. We have already reviewed, around equation (\ref{TT_field_dep_diff}), that the $\TT$ flow admits a beautiful interpretation in terms of a field-dependent metric involving the energy-momentum tensor \cite{Conti_2019}; subsequent work has extended this analysis in various directions, including stress tensor flows in higher dimensions \cite{Conti:2022egv,Morone:2024ffm,Ran:2024vgl,Ran:2025xas}, $T \overbar{T}_s$ deformations involving both the stress tensor and higher-spin currents \cite{Conti:2019dxg}, root-$\TT$ and combined $\TT$ plus root-$\TT$ flows \cite{Tsolakidis:2024wut,Babaei-Aghbolagh:2024hti,He:2025fdz}, curved spacetimes \cite{Caputa:2020lpa}, and connections with Ricci flow \cite{Brizio:2024arr,Morone:2024sdg}. One might ask whether there exists a similar metric interpretation for other auxiliary field deformations of $2d$ theories. In this section, we will explain that the answer is affirmative for general symmetric space sigma models (including $\mathbb{CP}^{N-1}$) deformed by auxiliary field couplings involving an interaction function of one variable, which includes arbitrary stress tensor flows. More precisely, we will show that

\begin{enumerate}[label=(\alph*)]
    \item\label{field_dep_metric} Every stress tensor deformation of a SSSM has equations of motion which are equivalent to those of the undeformed theory coupled to a unit-determinant field-dependent metric $h_{\mu \nu} ( j )$; here by ``unit-determinant'' we mean $\det ( h_{\mu \nu} ) = \det ( \eta_{\mu \nu} )$.\footnote{In our conventions for light-cone coordinates, $\eta_{+-} = \eta_{-+} = - 2$, one actually has $\det ( \eta_{\mu \nu} ) = - 4$ rather than $-1$, so the word ``unit'' is somewhat misleading.}

    \item\label{moving_frame} Any such auxiliary field deformation of a SSSM induces a particular linear transformation on the tangent space to its soliton surface, which can be viewed as a non-trivial choice of moving frame that nonetheless preserves the classical structure equations (equivalently, the integrability of the theory).
\end{enumerate}

Although these results are essentially restatements of previous observations in different language, we find them illuminating for several reasons. First, we will see that the field-dependent metric of point \ref{field_dep_metric} collapses to the usual flat metric on any field configurations with vanishing stress tensor, which offers another way to see the non-deformation of instanton solutions proved in Section \ref{s:instpres1}. Second, \ref{field_dep_metric} is closely connected to the fact that every $4d$ theory of duality-invariant electrodynamics can be viewed as the Maxwell theory coupled to a field-dependent metric \cite{Ferko:2024yhc}; this may not come as a surprise, since families of duality-invariant electrodynamics theories all obey stress tensor flows \cite{Ferko:2023wyi} and can be realized by the Ivanov-Zupnik \cite{Ivanov:2002ab,Ivanov:2003uj} auxiliary field formalism that is quite similar to the corresponding auxiliary field couplings for sigma models.\footnote{See also \cite{Ferko:2024zth,Babaei-Aghbolagh:2025lko,Babaei-Aghbolagh:2025cni,Babaei-Aghbolagh:2025uoz} for related work.} Third, points \ref{field_dep_metric} and \ref{moving_frame} are related: the pullback of the Cartan-Killing form along the Sym-Tafel embedding function encodes dynamical data about the field theory, such as the Lagrangian and stress tensor in the example (\ref{pullback_killing_cartan}). Thus, it is clear that deforming the theory's dynamics by introducing field dependence in the worldsheet metric must also modify the embedding map (and its derivatives, which define the tangents to the surface and hence the pullback) in a non-trivial way. Collectively, this pair of observations therefore gives a clear geometrical interpretation for the deformation of the soliton surface which is enacted by auxiliary field deformations and whose properties are studied in Sections \ref{s:cpnaux}. Finally, just as other deformations of Riemannian manifolds such as Ricci flow are of deep mathematical interest, as a future direction one might hope that the deformations \ref{moving_frame} of soliton surfaces associated with worldsheet auxiliary field deformations may also find applications in mathematics.

\subsection{Field-dependent metrics}\label{sec:field_dep_metric}

We work with a generic symmetric space sigma model, which following the definitions of Section \ref{ss:auxrev}, is the theory of a coset-valued field $g : \Sigma \to G / H$ where the orthogonal decomposition $\mathfrak{g} = \mathfrak{g}^{(0)} \oplus \mathfrak{g}^{(2)}$ of the Lie algebra introduced in (\ref{lie_algebra_decomp}) obeys (\ref{eq:cosetrels}) which is the condition that defines a \emph{symmetric} coset.

To be clear, there are two natural notions that one might mean by ``coupling'' a symmetric space sigma model to a field-dependent metric. The first is to begin with the undeformed action (\ref{undeformed_coset_action}) and replace the flat metric $\eta_{\mu \nu}$ with a field-dependent metric\footnote{This field-dependent metric $h_{\mu \nu} ( j )$ is not to be confused with the second fundamental form $h_{\mu \nu}$ of a soliton surface used, for instance, in equation (\ref{second_fundamental}).} $h_{\mu \nu} = h_{\mu \nu} ( j )$, with inverse $\left( h^{-1} \right)^{\mu \nu}$, that depends on the Maurer-Cartan form $j_\mu$ in some way:
\begin{equation}\label{field_dep_action}
    S [ g , h ] = \frac{1}{2} \int_\Sigma \ttr {j_\mu^{(2)}j_\nu^{(2)}} \left( h^{-1} \right)^{\mu\nu} \, d^2x \, .
\end{equation}
The second natural notion of coupling to a field-dependent metric is to instead begin with the equation of motion (\ref{undeformed_coset_eom}), which implicitly depends on the flat inverse metric $\eta^{\mu \nu}$ which is needed to raise an index (as both $j_\mu^{(2)}$ and the covariant derivative $\mathcal{D}_\mu$ are naturally defined with downstairs indices), and incorporate field dependence here:
\begin{align}\label{field_dep_eom}
    \mathcal{D}_\mu \left( \left( h^{-1} \right)^{\mu \nu} j^{(2)}_\nu \right) = 0 \, .
\end{align}
When we speak of coupling to a field-dependent metric, we will always intend the second of these two possible definitions, i.e., introducing field-dependence in the equation of motion (\ref{field_dep_eom}). The two definitions do not agree in general, since varying the fundamental group-valued field $g$ in the action (\ref{field_dep_action}) produces additional terms from the fluctuation of $\left( h^{-1} \right)^{\mu\nu}$.

We mention in passing that, by an argument essentially identical to that of \cite{Ferko:2024yhc}, the two notions (\ref{field_dep_action}) and (\ref{field_dep_eom}) happen to agree for the special case of the root-$\TT$ deformation of the SSSM, which corresponds to an auxiliary field model with interaction function (\ref{rtt_interaction}) proportional to $\sqrt{\nu_2}$. However, in the remainder of this section we will be interested in deformed models defined by an action (\ref{eq:auxacn}) where the interaction function takes the form (\ref{eq:cpE}) which is an arbitrary function of $\nu_2$ but, for simplicity, no other trace structures.\footnote{We stress, however, that thanks to \eqref{tracescpn}, this ansatz includes all possible higher-spin deformations by auxiliary fields of the $\Cp {N-1}$ model.}
Thus we will henceforth ignore the action-based coupling (\ref{field_dep_action}) and focus on  (\ref{field_dep_eom}).

As we have seen, when the auxiliary field equations of motion are satisfied, the equation of motion (\ref{eq:onshelleom}) for the physical field $g$ in the deformed symmetric space sigma model is
\begin{align}\label{covariant_deformed_eom}
    \mathcal{D}_\mu \mathfrak{J}^{(2) \mu} = 0 \, .
\end{align}
Comparing (\ref{field_dep_eom}) to (\ref{covariant_deformed_eom}), we see that one should identify
\begin{align}\label{J_field_dep_metric}
    \mathfrak{J}^{(2) \mu} = \left( h^{-1} \right)^{\mu \nu} j_\nu^{(2)} \, .
\end{align}
We will always raise or lower indices with the flat Minkowski metric $\eta_{\mu \nu}$ rather than with the field-dependent metric $h_{\mu \nu}$. For instance, lowering an index on (\ref{J_field_dep_metric}) gives
\begin{align}
    \mathfrak{J}^{(2)}_\mu = \eta_{\mu \rho} \left( h^{-1} \right)^{\rho \nu} j_\nu^{(2)} = \tensor{\left( h^{-1} \right)}{_\mu^\nu} j_\nu^{(2)} \, .
\end{align}
To avoid confusion, we will also introduce a new symbol for this quantity:
\begin{align}
    \tensor{R}{_\mu^\nu} = \eta_{\mu \rho} \left( h^{-1} \right)^{\rho \nu} \, .
\end{align}
We see that the ``unit-determinant'' condition $\det ( h_{\mu \nu} ) = \det ( \eta_{\mu \nu} )$ is equivalent to
\begin{align}\label{unimodular_Lambda}
    \det \left( R \right) = 1 \, .
\end{align}
Of course, the covariant conservation equation (\ref{covariant_deformed_eom}) is not the only condition that must be satisfied by fields in the deformed SSSM. One also has the Maurer-Cartan identity, whose projections (\ref{mc_projection_first_time}) onto the $\mathfrak{g}^{(0)}$ and $\mathfrak{g}^{(2)}$ subspaces we repeat for convenience:
\begin{align}\label{SSSM_MC}
    F_{\mu \nu}^{(0)} + [ j_\mu^{(2)} , j_\nu^{(2)} ] = 0 \, , \qquad \mathcal{D}_\mu j_\nu^{(2)} - \mathcal{D}_\nu j_\mu^{(2)} = 0 \, .
\end{align}
However, both of these equations only involve differential forms, commutators, and the covariant derivative $\mathcal{D}_\mu$, none of which involve the metric in their definition. Therefore, replacing $\eta^{\mu \nu}$ with the field-dependent inverse metric $\left( h^{-1} \right)^{\mu \nu}$ has precisely the effect of promoting the SSSM equation of motion to the AF-SSSM equation of motion while leaving the other conditions (\ref{SSSM_MC}) unchanged, exactly as desired.

The unimodularity condition (\ref{unimodular_Lambda}), along with the symmetry of the inverse metric $\left( h^{-1} \right)^{\mu \nu} = \left( h^{-1} \right)^{\nu \mu}$, are closely related to the commutator identities (\ref{commutator_identities}) which were used in \cite{Bielli:2024oif} to prove the classical integrability of the auxiliary field deformed SSSM. Indeed, it is trivial to show the following simple fact.

\begin{lem}
    Let $M^{\mu \nu}$ be any symmetric tensor field with $\det ( M^{\mu \nu} ) = \det ( \eta^{\mu \nu} ) = - \frac{1}{4}$, let $\tensor{R}{_\mu^\nu} = \eta_{\mu \rho} M^{\rho \nu}$, and assume $\mathfrak{J}^{(2)}_\mu = \tensor{R}{_\mu^\nu} j_\nu^{(2)}$. Then the identities (\ref{commutator_identities}) are satisfied.
\end{lem}

\begin{proof}

As usual, we work in conventions where $\eta_{+-} = \eta_{-+} = - 2$. If 
\begin{align}
    M^{\mu \nu} = \begin{bmatrix} M^{++} & M^{+-} \\ M^{+-} & M^{--} \end{bmatrix}^{\mu \nu} \, , 
\end{align}
we see that
\begin{align}\label{lambda_down_up}
    \tensor{R}{_\mu^\nu} = - 2 \tensor{\begin{bmatrix} M^{+-} & M^{--} \\ M^{++} & M^{+-} \end{bmatrix}}{_\mu^\nu} \, .
\end{align}
It is convenient\footnote{This notation is somewhat misleading since $b, c$ carry Lorentz spin, but we suppress this for simplicity.} to define $a = - 2 M^{+-}$, $b = - 2 M^{--}$, $c = -2 M^{++}$ so that
\begin{align}\label{nice_lambda_thingy}
    R = \begin{bmatrix} a & b \\ c & a \end{bmatrix} \, , \qquad a^2 - b c = 1 \, ,
\end{align}
where the latter condition follows from $\det ( M^{\mu \nu} ) = \det ( \eta^{\mu \nu} ) = - \frac{1}{4}$. Note that $R$ is not a generic $SL(2, \mathbb{R})$ matrix but also a \emph{Toeplitz matrix}, i.e. one for which all of the diagonal entries are equal to the same quantity $a$. Also recall that all of the quantities $M^{\mu \nu}$, and thus $a, b, c$, can be generic functions of the spacetime coordinates $x^{\pm}$ and of the fields $j_{\pm}$, but we suppress this dependence here. The components of the matrix equation $\mathfrak{J}^{(2)}_\mu = \tensor{R}{_\mu^\nu} j_\nu^{(2)}$ are
\begin{align}
    \mathfrak{J}_+^{(2)} = a j_+^{(2)} + b j_-^{(2)} \, , \qquad \mathfrak{J}_-^{(2)} = c j_+^{(2)} + a j_-^{(2)} \, .
\end{align}
Clearly
\begin{align}
    [ \mathfrak{J}_+^{(2)} , \mathfrak{J}_-^{(2)} ] &= \left[ a j_+^{(2)} + b j_-^{(2)} ,  c j_+^{(2)} + a j_-^{(2)} \right] \nonumber \\
    &= \left( a^2 - b c \right) [ j_+^{(2)} , j_-^{(2)} ] \nonumber \\
    &= [ j_+^{(2)} , j_-^{(2)} ] \, ,
\end{align}
where we have used the unit-determinant condition $a^2 - b c = 1$. Likewise
\begin{align}
    [ \mathfrak{J}_+^{(2)} , j_-^{(2)} ] &= \left[ a j_+^{(2)} + b j_-^{(2)} , j_-^{(2)} \right] = a [ j_+^{(2)} , j_-^{(2)} ] \, ,
\end{align}
while
\begin{align}
    [ j_+^{(2)} , \mathfrak{J}_-^{(2)} ] &= \left[ j_+^{(2)} , c j_+^{(2)} + a j_-^{(2)} \right] = a [ j_+^{(2)} , j_-^{(2)} ] \, ,
\end{align}
so $[ \mathfrak{J}_+^{(2)} , j_-^{(2)} ] = [ j_+^{(2)} , \mathfrak{J}_-^{(2)} ]$.
\end{proof}

Thus we see that the relation (\ref{J_field_dep_metric}), where crucially the field-dependent metric is unit-determinant, underlies the integrability of the deformed symmetric space sigma model. Next let us show that such a unit-determinant field-dependent metric exists for all auxiliary field deformed SSSMs with interaction functions (\ref{eq:cpE}) depending only on $\nu_2$. Combining the constraint (\ref{eq:veomcp1}) from the auxiliary field equation of motion with the definition
\begin{align}
    \mathfrak{J}_\pm^{(2)} = - \left( j_\pm^{(2)} + 2 v_\pm^{(2)} \right)
\end{align}
obtained from projecting (\ref{frakJ_defn}) onto $\mathfrak{g}^{(2)}$ and solving, we find
\begin{gather}\label{R_matrix_defn}
    \begin{bmatrix} \mathfrak{J}_+^{(2)} \\ \mathfrak{J}_-^{(2)} \end{bmatrix} = R \begin{bmatrix} j_+^{(2)} \\ j_-^{(2)} \end{bmatrix} \, , \\
    R = \frac{1}{1 - 4 ( E' )^2 \nu_2} \begin{bmatrix} 1 + 4 ( E' )^2 \nu_2 & - 4 E' \tr ( v_+^{(2)} v_+^{(2)} ) \\ - 4 E' \tr ( v_-^{(2)} v_-^{(2)} )  & 1 + 4 ( E' )^2 \nu_2 \end{bmatrix} \, ,
\end{gather}
where again $\nu_2 = \tr ( v_+^{(2)} v_+^{(2)} ) \tr ( v_-^{(2)} v_-^{(2)} )$. Defining
\begin{align}\label{abc_defn}
    a = \frac{1 + 4 ( E' )^2 \nu_2}{1 - 4 ( E' )^2 \nu_2} \, , \qquad b = - \frac{4 E' \tr ( v_+^{(2)} v_+^{(2)} )}{1 - 4 ( E' )^2 \nu_2} \, , \qquad c = - \frac{4 E' \tr ( v_-^{(2)} v_-^{(2)} )}{1 - 4 ( E' )^2 \nu_2} \, ,
\end{align}
we see that
\begin{align}
    a^2 - b c &= \frac{1}{ \left( 1 - 4 ( E' )^2 \nu_2 \right)^2 } \left( \left( 1 + 4 ( E' )^2 \nu_2 \right)^2 - 16 ( E' )^2 \nu_2 \right) \nonumber \\
    &= \frac{1}{ \left( 1 - 4 ( E' )^2 \nu_2 \right)^2 } \left( 1 - 8 ( E' )^2 \nu_2 + 16 ( E' )^4 \nu_2^2 \right) \nonumber \\
    &= 1 \, .
\end{align}
Thus every such model obeys a relation $\mathfrak{J}^{(2)}_\mu = \tensor{R}{_\mu^\nu} j_\nu^{(2)}$ where $R$ takes the form (\ref{nice_lambda_thingy}) of a Toeplitz $SL(2, \mathbb{R})$ matrix. Again, the entries of this matrix are field-dependent, since they contain the auxiliary fields $v_\pm^{(2)}$; in principle, these fields can always be integrated out and expressed in terms of the physical fields $j_\pm^{(2)}$ using the auxiliary field equation of motion, although this can be analytically difficult in practice. 

Note that the quantity in the denominators of (\ref{abc_defn}), $1 - 4 ( E' )^2 \nu_2$, vanishes only for $E ( \nu_2 ) = \pm \sqrt{ \nu_2 }$. This corresponds to the root-$\TT$ interaction function (\ref{rtt_interaction}) in the $\gamma \to \pm \infty$ limit. We exclude this pathological case from our consideration because it corresponds to a strictly infinite deformation parameter. For all other choices of interaction function, the entries in the matrix $R$ of (\ref{R_matrix_defn}) are well-defined.

Raising an index to express this result in terms of $\left( h^{-1} \right)^{\mu \nu}$ as in (\ref{J_field_dep_metric}), rather than in terms of $R$, we see that
\begin{align}
    \left( h^{-1} \right)^{\mu \nu} = \frac{1}{1 - 4 ( E' )^2 \nu_2} \begin{bmatrix} 2 E' \tr ( v_-^{(2)} v_-^{(2)} ) & - \frac{1}{2} - 2 ( E' )^2 \nu_2 \\ - \frac{1}{2} - 2 ( E' )^2 \nu_2  & 2 E' \tr ( v_+^{(2)} v_+^{(2)} ) \end{bmatrix}^{\mu \nu} \, ,
\end{align}
or taking the inverse,
\begin{align}\label{field_dep_metric_downstairs}
    h_{\mu \nu} =  \frac{1}{1 - 4 ( E' )^2 \nu_2} \begin{bmatrix} - 8 E' \tr ( v_+^{(2)} v_+^{(2)} ) & - 2 - 8 ( E' )^2 \nu_2 \\ - 2 - 8 ( E' )^2 \nu_2 & - 8 E' \tr ( v_-^{(2)} v_-^{(2)} ) \end{bmatrix}_{\mu \nu} \, .
\end{align}
As expected, one finds $\det ( h_{\mu \nu} ) = - 4$ and $\det ( \left( h^{-1} \right)^{\mu \nu} ) = - \frac{1}{4}$.

This completes the proof of claim \ref{field_dep_metric} at the beginning of this section: every auxiliary field deformation of a SSSM with an interaction function $E = E ( \nu_2 )$ has equations of motion which are equivalent to those of the seed theory on a field-dependent metric $h_{\mu \nu}$ with $\det ( h ) = \det ( \eta )$. As a conceptual application of this result, we conclude by presenting another perspective on the non-deformation of instantons described in Section \ref{s:instpres1}.

\begin{lem}\label{instanton_flat_metric_lemma}
    Consider any instanton solution of the symmetric space sigma model, i.e. one for which the energy-momentum tensor is vanishing. Evaluating the field-dependent metric $h_{\mu \nu}$ of (\ref{field_dep_metric_downstairs}) on such a solution produces $h_{\mu \nu} = \eta_{\mu \nu}$.
\end{lem}

\begin{proof}

In the undeformed SSSM, the stress tensor components $T_{\pm \pm}$ vanish if and only if the quantities $\tr ( j_\pm^{(2)} j_\pm^{(2)} )$ vanish, so we assume the latter. Solving the constraint (\ref{eq:veomcp1}) for $j_\pm^{(2)}$ and squaring it gives
\begin{align}
    j_\pm^{(2)} j_\pm^{(2)} &= \left( v_\pm^{(2)} + 2 E' \tr ( v_\pm^{(2)} v_\pm^{(2)} ) v_\mp^{(2)} \right) \left( v_\pm^{(2)} + 2 E' \tr ( v_\pm^{(2)} v_\pm^{(2)} ) v_\mp^{(2)} \right) \nonumber \\
    &= v_\pm^{(2)} v_\pm^{(2)} + 2 E' ( v_\pm^{(2)} v_\mp^{(2)} + v_\mp^{(2)} v_\pm^{(2)} ) \tr ( v_\pm^{(2)} v_\pm^{(2)} ) + 4 ( E' )^2 v_\mp^{(2)} v_\mp^{(2)} \left( \tr ( v_\pm^{(2)} v_\pm^{(2)} ) \right)^2 \, .
\end{align}
Taking the trace of this equation then gives
\begin{align}
    \tr ( j_\pm^{(2)} j_\pm^{(2)} ) &= \tr ( v_\pm^{(2)} v_\pm^{(2)} ) + 4 E' \tr ( v_\pm^{(2)} v_\mp^{(2)} ) \tr ( v_\pm^{(2)} v_\pm^{(2)} ) + 4 ( E' )^2 \tr ( v_\mp^{(2)} v_\mp^{(2)} ) \left( \tr ( v_\pm^{(2)} v_\pm^{(2)} ) \right)^2 \nonumber \\
    &= \tr ( v_\pm^{(2)} v_\pm^{(2)} ) \left( 1 + 4 E' \tr ( v_\pm^{(2)} v_\mp^{(2)} ) + 4 (  E' )^2 \tr ( v_\mp^{(2)} v_\mp^{(2)} ) \tr ( v_\pm^{(2)} v_\pm^{(2)} ) \right) \, .
\end{align}
Thus, for a field configuration in which the two equations $\tr ( j_\pm^{(2)} j_\pm^{(2)} ) = 0$ are satisfied, we must have either the two equations
\begin{align}
    \tr ( v_\pm^{(2)} v_\pm^{(2)} ) = 0 \, ,
\end{align}
or else the single equation
\begin{align}\label{wrong}
    1 + 4 E' \tr ( v_+^{(2)} v_-^{(2)} ) + 4 (  E' )^2 \tr ( v_+^{(2)} v_+^{(2)} ) \tr ( v_-^{(2)} v_-^{(2)} ) = 0 \, .
\end{align}
However, equation (\ref{wrong}) is a quadratic relation that allows us to solve for $E'$ in terms of the two variables $\nu_2 = \tr ( v_+^{(2)} v_+^{(2)} ) \tr ( v_-^{(2)} v_-^{(2)} )$ and $\tr ( v_+^{(2)} v_-^{(2)} )$. But by assumption, $E$ and thus $E'$ depend only on the single variable $\nu_2$ and \emph{not} on $\tr ( v_+^{(2)} v_-^{(2)} )$. Therefore (\ref{wrong}) is a contradiction and we reject this equation.

We conclude that $\tr ( v_\pm^{(2)} v_\pm^{(2)} ) = 0$ and hence $\nu_2 = 0$; evaluating (\ref{field_dep_metric_downstairs}) in this case gives
\begin{align}
    h_{\mu \nu} = \begin{bmatrix} 0 & -2 \\ -2 & 0 \end{bmatrix}_{\mu \nu} \, ,
\end{align}
which is simply $h_{\mu \nu} = \eta_{\mu \nu}$ in our light-cone conventions.
\end{proof}

This gives a satisfying interpretation for the non-deformation of instantons, which satisfy the on-shell vanishing condition (\ref{on_shell_T_vanish}) of the energy-momentum tensor. We have already seen that any stress tensor deformation of a SSSM -- that is, any auxiliary field deformation with an interaction function depending only on $\nu_2$ -- can be interpreted as the undeformed SSSM on a field-dependent metric. However, for those specific field configurations whose stress tensors vanish, Lemma \ref{instanton_flat_metric_lemma} shows that this field-dependent metric precisely reduces to the flat undeformed metric. It follows that any solution of the undeformed theory with a vanishing energy-momentum tensor is also a solution of the deformed model. 

We therefore see geometrically why instanton solutions are undeformed, as proved in Section \ref{s:instpres1}, from the complementary metric perspective. Although we have presented these results for symmetric space sigma models, they also immediately apply to the principal chiral model with target Lie group $G$ (for instance, by taking a coset $( G \times G ) / G$), and we expect that the metric interpretation can be extended to auxiliary field deformations of other coset theories \cite{Cesaro:2024ipq,Bielli:2024oif}, the Yang-Baxter deformed PCM \cite{Bielli:2024fnp}, and so on. Let us now turn to the geometrical interpretation of the other primary focus of this work, namely the soliton surfaces of deformed $\mathbb{CP}^{N-1}$ (or more general symmetric coset) models.

\subsection{Transformation of soliton surface}

We have reviewed, in Section \ref{sec:generalities_soliton}, how the integrable structure of $2d$ field theories can be reinterpreted as statements about the soliton surface defined by the Sym-Tafel immersion $r$ of equation (\ref{sym_tafel}): equality of mixed second derivatives is implied by integrability, periods of the soliton surface are logarithmic $z$-derivatives of the monodromy matrix, and so on. 

In particular, for the $\mathbb{CP}^{N-1}$ model, we have seen in equation (\ref{pullback_killing_cartan}) that the metric components $g_{\pm \pm}$ on the soliton surface are related to the energy-momentum tensor $T_{\pm \pm}$, while $g_{\pm \mp} = g_{\mp \pm}$ is related to the Lagrangian $\mathcal{L}$ (in fact, this is true for more general sigma models). Stress tensor flows, by definition, deform the Lagrangian $\mathcal{L}$ by a function of the energy-momentum tensor. It is therefore natural to expect that such deformations can be interpreted as a modification of the metric on the soliton surface, as we have seen in equation (\ref{deformed_soliton_metric}). Our goal in the present section is to give a geometrical interpretation of this modification in terms of a particular linear transformation on the tangent space to the soliton surfaces, as alluded to in point \ref{moving_frame} above and which is related to the field-dependent metric of statement \ref{field_dep_metric} and developed in the preceding subsection \ref{sec:field_dep_metric}.

We begin by noting that the tangent vectors $r_\mu = \partial_\mu r$ to the soliton surface, computed in equation (\ref{tangent_computation}), can be converted to ``fixed-body'' tangent vectors $t_\mu$ defined by
\begin{align}
    t_\mu = \mathrm{Ad}_\Phi r_\mu = \partial_z L_\mu \, .
\end{align}
In terms of these quantities, equation (\ref{pullback_metric}) becomes $g_{\mu \nu} = \langle r_\mu , r_\nu \rangle = \langle t_\mu , t_\nu \rangle$, due to the Ad-invariance of the bilinear form $\langle \, \cdot \, , \, \cdot \rangle$. Therefore, for questions involving the metric on the soliton surface, it is sufficient to work with the fixed-body tangents $t_\mu$.

Working in light-cone coordinates and taking the $z$-derivatives of the components of the AF-SSSM Lax connection (\ref{eq:defLax}), we see that the $j^{(0)}_\pm$ contributions completely drop out:
\begin{align}
    t_+^{\text{AF-SSSM}} = - 2 \frac{ ( 1 + z^2 ) \mathfrak{J}_+^{(2)} - 2 z j_+^{(2)}}{ ( z^2 - 1 )^2 } \, , \qquad t_-^{\text{AF-SSSM}} = 2 \frac{ ( 1 + z^2 ) \mathfrak{J}_-^{(2)} + 2 z j_-^{(2)}}{ ( z^2 - 1 )^2 } \, .
\end{align}
We decorate these free-body tangents with the superscript AF-SSSM to distinguish them from the corresponding objects in the undeformed SSSM,
\begin{align}\label{old_fixed_tangents}
    t_+^{\text{SSSM}} = - 2 \frac{ j_+^{(2)} }{ ( z + 1 )^2} \, , \qquad t_-^{\text{SSSM}} = 2 \frac{j_-^{(2)}}{ ( z - 1 )^2} \, .
\end{align}
We now wish to relate these two sets of tangents. Recall that
\begin{align}
    \mathfrak{J}_+^{(2)} = a j_+^{(2)} + b j_-^{(2)} \, , \qquad \mathfrak{J}_-^{(2)} = c j_+^{(2)} + a j_-^{(2)} \, ,
\end{align}
for $a, b, c$ given in (\ref{abc_defn}). Solving then yields
\begin{align}
    t_+^{\text{AF-SSSM}} &= \frac{a - 2 z + a z^2 }{ ( 1 - z )^2} t_+^{\text{SSSM}} - \frac{b ( 1 + z^2 )}{ ( z + 1 )^2 } t_-^{\text{SSSM}} \, , \nonumber \\
    t_-^{\text{AF-SSSM}} &= - \frac{c ( 1 + z^2 )}{ ( 1 - z )^2}  t_+^{\text{SSSM}} + \frac{a + 2 z + a z^2}{ ( z + 1 )^2} t_-^{\text{SSSM}} \, .
\end{align}
We thus have
\begin{align}\label{fixed_tangent_transformation}
    t_\mu^{\text{AF-SSSM}} = \tensor{\widetilde{\Lambda}}{_\mu^\nu} t_\nu^{\text{SSSM}} \, ,
\end{align}
where the explicit entries in the matrix $\widetilde{\Lambda}$ in light-cone coordinates are
\begin{align}\label{full_Lambda_substituted}
    \tensor{\widetilde{\Lambda}}{_\mu^\nu} = \frac{1}{1 - 4 ( E' )^2 \nu_2 } \tensor{\left( \begin{bmatrix} 1 & 0 \\ 0 & 1 \end{bmatrix} + 4 E' \begin{bmatrix} E' \nu_2 \frac{ ( 1 + z )^2}{ ( 1 - z )^2} & \tr ( v_+^{(2)} v_+^{(2)} ) \frac{1 + z^2}{ ( 1 + z )^2 } \\ \tr ( v_-^{(2)} v_-^{(2)} ) \frac{1 + z^2}{ ( 1 - z )^2 } & E' \nu_2  \frac{( 1 - z )^2}{( 1 + z )^2 } \end{bmatrix} \right)}{_\mu^\nu} \, .
\end{align}
We have therefore found that the full effect of the auxiliary field deformation of the SSSM is to enact a particular linear transformation on the fixed-body tangent vectors to the surface. 

A few comments are in order. First, just as the metric $h_{\mu \nu}$ introduced in Section \ref{sec:field_dep_metric} was field-dependent, the entries in the matrix $\widetilde{\Lambda}$ depend on the auxiliary fields $v_\pm^{(2)}$ or -- after integrating out these fields using their algebraic equations of motion -- on $j_\pm^{(2)}$, which are themselves proportional to the undeformed fixed-body tangent vectors (\ref{old_fixed_tangents}). Therefore, the linear transformation (\ref{fixed_tangent_transformation}) is highly non-trivial as the matrix elements themselves depend on the undeformed fixed-body tangents $t_\pm^{\text{SSSM}}$.

Second, recall that a similar structure was observed in equation (\ref{other_lax_lambda}) relating $L^{(2)}$ to its formal expression upon setting $E = 0$ by a linear transformation $\Lambda$. We have used a different symbol $\widetilde{\Lambda}$ for the matrix here to emphasize that they are different transformations with different interpretations. The matrix $\Lambda$ acts upon the ``undeformed Lax'' (\ref{eq:undefLx}) which is technically evaluated on solutions of the deformed theory upon formally setting $E = 0$. In contrast, the matrix $\widetilde{\Lambda}$ relates the fixed-body tangents, i.e. the $z$-derivatives of the Lax, in the true undeformed theory to those in the true deformed theory. As a result, the metric of the soliton surface for the deformed theory,
\begin{align}
    g_{\mu \nu}^{\text{AF-SSSM}} = \langle t_\mu^{\text{AF-SSSM}} , t_\nu^{\text{AF-SSSM}} \rangle = \left( t_\mu^A \right)^{\text{AF-SSSM}} \left( t_\nu^B \right)^{\text{AF-SSSM}} \gamma_{AB} \, ,
\end{align}
where we have introduced the components $\gamma_{AB}$ of the Cartan-Killing form, obeys
\begin{align}
    g_{\mu \nu}^{\text{AF-SSSM}} &= \tensor{\widetilde{\Lambda}}{_\mu^\rho} \tensor{\widetilde{\Lambda}}{_\nu^\sigma} \left( t_\rho^A \right)^{\text{SSSM}} \left( t_\sigma^B \right)^{\text{SSSM}} \gamma_{AB} \nonumber \\
    &= \tensor{\widetilde{\Lambda}}{_\mu^\rho} g_{\rho \sigma}^{\text{SSSM}} \tensor{\widetilde{\Lambda}}{_\nu^\sigma} \, .
\end{align}
This is reminiscent of the relation
\begin{align}
    g_{\mu \nu} = e_\mu^a \eta_{ab} e_\nu^b \, ,
\end{align}
between a tangent space metric $\eta_{ab}$ and a spacetime metric $g_{\mu \nu}$ which are related by a choice of vielbein or frame $e_\mu^a$. Thus the effect of the auxiliary field deformation is to make a particular choice of ``moving frame'' $e_\mu^a$ or $\tensor{\widetilde{\Lambda}}{_\mu^\nu}$ on the soliton surface, which varies from point to point on $\Sigma$ according to its implicit dependence on the fields $v_\pm ( x )$ or $j_\pm ( x )$. This extends the result (\ref{metric_Lambda_relation}), which relates the deformed metric to the formally undeformed metric (obtained by setting $E = 0$) by a linear transformation, although again the transformations $\Lambda$ and $\widetilde{\Lambda}$ differ since the quantity $g_{\mu \nu} \vert_{E = 0}$ of (\ref{metric_Lambda_relation}) is not the same as $g_{\mu \nu}^{\text{SSSM}} $.

Third, although we have studied the relation between fixed-body tangents $t_\mu = \mathrm{Ad}_\Phi r_\mu$, it is simple to rephrase the result in terms of the standard tangent vectors $r_\mu$ by inserting appropriate factors of $\mathrm{Ad}$:
\begin{align}
    r_\mu^{\text{AF-SSSM}} = \mathrm{Ad}_{\Phi^{\text{AF-SSSM}}} \left( \tensor{\widetilde{\Lambda}}{_\mu^\nu}  \mathrm{Ad}_{\Phi^{\text{SSSM}}}^{-1}  \left( r_\mu^{\text{SSSM}} \right) \right) \, ,
\end{align}
where $L_\mu^{\text{SSSM}} \Phi^{\text{SSSM}} = \partial_\mu \Phi^{\text{SSSM}}$ and $L_\mu^{\text{AF-SSSM}} \Phi^{\text{AF-SSSM}} = \partial_\mu \Phi^{\text{AF-SSSM}}$, respectively, and the adjoints are understood to act separately on each component of the vectors.

Finally, it may seem that this geometrical deformation is trivial, since it simply corresponds to a choice of frame $e_\mu^a$ in some sense. But we emphasize again that the components of the matrix $\widetilde{\Lambda}$ are field-dependent, and thus depend on undeformed tangent vectors. Performing a change of frame involving an \emph{arbitrary} such matrix $\widetilde{\Lambda}$ would generically violate consistency conditions on the soliton surface, such as the equality of mixed partial derivatives (\ref{mixed_second_equal}). This is clear since such consistency conditions arise from the integrability of the model (i.e. they are implied by the on-shell flatness of the Lax connection), and a generic deformation of a $2d$ integrable theory does not preserve integrability. The surprising feature of this geometrical deformation is that, despite the complicated nature of the choice of moving frame, integrability is preserved (as shown in \cite{Bielli:2024oif}), and therefore this non-trivial linear transformation (\ref{full_Lambda_substituted}) still leads to a well-defined surface satisfying the appropriate structure equations (the Gauss equation, the Codazzi-Mainardi equation, and the Ricci equation of the normal bundle described in Section \ref{sec:generalities_soliton}).

\section{Conclusion and outlook}\label{sec:conclusion}

In this paper, we studied soliton surfaces of the $\mathbb{CP}^{N-1}$ model and its $\TTbar$-like deformation. We constructed soliton surfaces of the model using the Sym-Tafel formula and showed that important physical quantities like the Lagrangian and the energy-momentum tensor are encoded geometrically in the metric tensor of the soliton surface. We then focused on the $\mathbb{CP}^{1}$ case and computed the Gaussian and mean curvatures, finding that the Gaussian curvature is constant -- consistent with the results obtained using other immersion methods in \cite{grundland2012soliton}.\footnote{Although it has been argued \cite{grundland2009analytic} that different immersion approaches are equivalent for the $\mathbb{CP}^{N-1}$ model, readers interested in alternative immersions may refer to \cite{Sym:1981an,grundland2012soliton,grundland2016immersionformulassolitonsurfaces,Grundland2016generalizedsymmetryapproach,Grundland_2011integrableequations} for further details.} Since the $\mathbb{CP}^{N-1}$ model involves a unit constraint, we examined whether this constraint would be affected by the deformation. The result is negative: even though the deformation can potentially introduce an enhanced constraint, it contradicts the seed theory and therefore should not be adopted. Since the deformation approach we applied \cite{Ferko:2022cix} is valid for general sigma models, this indicates that no enhanced constraint should be adopted for any $\TTbar$-deformed sigma model when the undeformed constraints are holonomic. We showed that the topological chiral (instanton) solutions are preserved under the $\TTbar$ deformation and extended this result by providing a general proof that, for any system with a solution whose energy-momentum tensor vanishes, such a solution remains valid in any analytic $\TTbar$-like deformation of that system. Following that, we presented an alternative formalism for the $\mathbb{CP}^{N-1}$ model, expressing it as a coset model, and focused on its higher-spin auxiliary field deformation, which encompasses cases such as $\TTbar$-like deformations and Smirnov-Zamolodchikov higher-spin deformations \cite{Ferko:2024ali}. Within this framework, we demonstrated that the constant Gaussian curvature of soliton surfaces in the $\mathbb{CP}^{1}$ model remains unaffected by any higher-spin deformation. We also showed that in the auxiliary field formalism, the preservation of the instanton solution becomes manifest and follows straightforwardly. We identified a natural sufficient condition that undeformed solutions must satisfy in order to remain valid along a general SZ-like flow for a broad class of sigma models, and extended the argument to the $\sqrt{T\overbar{T}}$ case. Finally, we presented two geometrical interpretations for stress tensor deformations of symmetric space sigma models: they can be viewed either as coupling the undeformed theory to a unit-determinant field-dependent metric, or by choosing a special moving frame on the soliton surface which preserves the classical structure equations that are guaranteed by the integrability of the theory.

Soliton surfaces have gained increasing importance in the study of two-dimensional integrable models in recent years, as they offer a promising means of transforming problems in differential equations into geometric ones, thereby providing a natural geometric interpretation \cite{Sym:1981an,Sym:1983px,solitonsurfacesandtheirapplications,https://doi.org/10.1002/sapm19969619,Bobenko1994,CIESLINSKI19971,bobenko2000painleve,helein2001constant,grundland2012soliton,grundland2016immersionformulassolitonsurfaces,Grundland2016generalizedsymmetryapproach,Grundland_2011integrableequations,Yesmakhanova_2019,Conti_2019,GLandolfi_2003,Baran_2010,talukdar2024fokaslenellsderivativenonlinearschrodinger,Bertrand_2016,Bertrand_2017}. Moving forward, it would be interesting to investigate how important physical quantities are geometrically encoded on soliton surfaces associated with other sigma models, and even with theories beyond the sigma model framework. Such exploration may shed light on the underlying structure of these surfaces and deepen our understanding of their geometric significance. Furthermore, while progress has been made in extending Lax pairs to supersymmetric settings \cite{Bertrand:2017hee,Saleem:2006tv,10.1063/1.528090,Napolitano:1982xg,Bertrand_2016,Bertrand_2017}, it remains an open and intriguing question whether the notion of soliton surfaces can be generalized to supersymmetric models as well. In addition, although we have shown in this paper that no enhanced constraint arises from the $\TTbar$ deformation for general sigma models with holonomic constraints, it is worth exploring whether this result still holds in other classes of models. Lastly, while our construction of soliton surfaces relied on the Sym–Tafel formula, there exist alternative immersion methods that may offer new perspectives, especially when applied to different models. Exploring these methods in parallel could lead to a more comprehensive geometric picture of integrable field theories \cite{grundland2012soliton,grundland2016immersionformulassolitonsurfaces,Grundland2016generalizedsymmetryapproach,Grundland_2011integrableequations,talukdar2024fokaslenellsderivativenonlinearschrodinger}.

\section*{Acknowledgements}

We are very grateful to the participants of the ``Deformations of Quantum Field and Gravity Theories'' mini-workshop, held at the University of Queensland (UQ) from January 30 - February 6, 2025 -- in particular, Nicola Baglioni, Daniele Bielli,  Jue Hou, Tommaso Morone, and Roberto Tateo -- 
for many productive discussions about the ideas in this paper and related topics. 
C.\,F. is supported by the National Science Foundation under Cooperative Agreement PHY-2019786 (the NSF AI Institute for Artificial Intelligence and Fundamental Interactions).
M.\,G. and G.\,T.-M. have been supported by the Australian Research Council (ARC) Future Fellowship FT180100353, ARC Discovery
Project DP240101409, the Capacity Building Package at the University of Queensland, and a faculty start-up funding of UQ's School of Mathematics and Physics.
Z.\,H.  is supported by a postgraduate scholarship at the University
of Queensland.

\appendix

\section{Checking the topological chiral solution under deformation} \label{Expcheck}

For the topological chiral solutions $D_{\pm}\phi=0$, we have the following identities:
\begin{subequations}
\begin{gather}
    D_{0}\phi=D_{\mp}\phi\quad D_{1}\phi=\mp iD_{\mp}\phi\\
    Y=4[(D_{\mp}\phi)^{\dagger}(D_{\mp}\phi)]\quad X=8[(D_{\mp}\phi)^{\dagger}(D_{\mp}\phi)]^2=\frac{Y^2}{2}\\
    Y^2-X=8[(D_{\mp}\phi)^{\dagger}(D_{\mp}\phi)]^2=X\\
    X_{01}=X_{10}=0\quad X_{00}=X_{11}=2[(D_{\mp}\phi)^{\dagger}(D_{\mp}\phi)]=\frac{Y}{2}.
\end{gather}
\label{instantoncase}
\end{subequations}
The equation of motion in Eq. \eqref{deformedEOM} can be expanded as
\begin{equation}
    \begin{split}
    &\quad\underbrace{\frac{2\gamma(\partial_{\mu}X^{\mu\nu})(D_{\nu}\phi_{i})-2\gamma(\partial_{\mu}Y)(D^{\mu}\phi_{i})}{\sqrt{1-2\gamma Y+2\gamma^2(Y^2-X)}}}_{A}\\
    &\underbrace{-\Big(2\gamma X^{\mu\nu}(D_{\nu}\phi_{i})+(1-2\gamma Y)(D^{\mu}\phi_{i})\Big)\frac{-\gamma(\partial_{\mu}Y)+\gamma^2\partial_{\mu}(Y^2-X)}{\Big(1-2\gamma Y+2\gamma^2(Y^2-X)\Big)^{\frac{3}{2}}}}_{B}\\
    &+\underbrace{\frac{2\gamma X^{\mu\nu}(D_{\nu}D_{\mu}\phi_{i})+(1-2\gamma Y )(D^2\phi_{i})}{\sqrt{1-2\gamma Y+2\gamma^2(Y^2-X)}}}_{C}\underbrace{-\frac{2\gamma X^{\mu\nu}[\phi^{\dagger}(\partial_{\mu}D_{\nu}\phi)]+(1-2\gamma Y)[\phi^{\dagger}(\partial_{\mu}D^{\mu}\phi)]}{\sqrt{1-2\gamma Y+2\gamma^2(Y^2-X)}}\phi_{i}}_{D}=0,
    \end{split}
\end{equation}
which consists of four main terms, labelled A through D.

We now rewrite all terms with the common denominator
\begin{equation}
 \Big(1-2\gamma Y+2\gamma^2(Y^2-X)\Big)^{\frac{3}{2}},
\end{equation}
and examine the behavior of the equation at different orders in $\gamma$.  If the topological chiral solutions satisfy the equations of motion, then the coefficient at each order of $\gamma$ must vanish independently. The $0^{\rm th}$-order term is simply the equation of motion in Eq. \eqref{EOMCPN-1}, which vanishes automatically when evaluated on any solution to the undeformed system. For the $\gamma^2$ contribution, we have
\begin{equation}
    \gamma^2\Big(-4Y(\partial_{\mu}X^{\mu\nu})(D_{\nu}\phi_{i})+4Y(\partial_{\mu}Y)(D^{\mu}\phi_{i})\Big)
\end{equation}
from A,
\begin{equation}
    \gamma^2\Big(2X^{\mu\nu}(\partial_{\mu}Y)(D_{\nu}\phi_{i})-\partial_{\mu}(Y^2-X)(D^{\mu}\phi_{i})-2Y(\partial_{\mu}Y)(D^{\mu}\phi_{i})\Big)
\end{equation}
from B,
\begin{equation}
    \gamma^2\Big(-4YX^{\mu\nu}(D_{\nu}D_{\mu}\phi_{i})+4Y^2(D^2\phi_{i})+2(Y^2-X)(D^2\phi_{i})\Big)
\end{equation}
from C, and
\begin{equation}
    \gamma^2\Big(4YX^{\mu\nu}[\phi^{\dagger}(\partial_{\mu}D_{\nu}\phi)]-4Y^2[\phi^{\dagger}(\partial_{\mu}D^{\mu}\phi)]-2(Y^2-X)[\phi^{\dagger}(\partial_{\mu}D^{\mu}\phi)]\Big)\phi_{i}
\end{equation}
from D. Gathering all terms and applying the identities in Eq. \eqref{instantoncase} to simplify, we obtain
\begin{equation}
    3Y^2\Big((D^2\phi_{i})-[\phi^{\dagger}(\partial_{\mu}D^{\mu}\phi)]\phi_{i}\Big)=3Y^2\Big((D^2\phi_{i})+(D^{\mu}\phi)^{\dagger}(D_{\mu}\phi)\phi_{i}\Big)=0,
\end{equation}
which confirms that the $\gamma^2$ term is indeed zero.

Similarly, the $\gamma^3$ contribution consists of
\begin{equation}
    \gamma^3(Y^2-X)\Big(4(\partial_{\mu}X^{\mu\nu})(D_{\nu}\phi_{i})-4(\partial_{\mu}Y)(D^{\mu}\phi_{i})\Big),
\end{equation}
from A,
\begin{equation}
    \gamma^3\partial_{\mu}(Y^2-X)\Big(-2X^{\mu\nu}(D_{\nu}\phi_{i})+2Y(D^{\mu}\phi_{i})\Big)
\end{equation}
from B,
\begin{equation}
    \gamma^3(Y^2-X)\Big(4X^{\mu\nu}(D_{\nu}D_{\mu}\phi_{i})-4Y(D^2\phi_{i})\Big)
\end{equation}
from C, and
\begin{equation}
    \gamma^3(Y^2-X)\Big(-4X^{\mu\nu}[\phi^{\dagger}(\partial_{\mu}D_{\nu}\phi)]\phi_{i}+4Y[\phi^{\dagger}(\partial_{\mu}D^{\mu}\phi)]\phi_{i}\Big)
\end{equation}
from D, and by repeating the same step as before, we eventually obtain
\begin{equation}
    X\Big(-2(\partial_{\mu}Y)(D^{\mu}\phi_{i})-2Y(D^2\phi_{i})+2Y[\phi^{\dagger}(\partial_{\mu}D^{\mu}\phi)]\phi_{i}\Big)+Y(\partial_{\mu}X)(D^{\mu}\phi_{i})=0.
\end{equation}
Finally, for the linear order in $\gamma$, we have
\begin{equation}
    \gamma\Big(2(\partial_{\mu}X^{\mu\nu})(D_{\nu}\phi_{i})-2(\partial_{\mu}Y)(D^{\mu}\phi_{i})\Big)
\end{equation}
from A,
\begin{equation}
    \gamma(\partial_{\mu}Y)(D^{\mu}\phi_{i})
\end{equation}
from B,
\begin{equation}
    \gamma\Big(2X^{\mu\nu}(D_{\nu}D_{\mu}\phi_{i})-4Y(D^2\phi_{i})\Big)
\end{equation}
from C, and
\begin{equation}
    \gamma\Big(-2X^{\mu\nu}[\phi^{\dagger}(\partial_{\mu}D_{\nu}\phi)]\phi_{i}+4Y[\phi^{\dagger}(\partial_{\mu}D^{\mu}\phi)]\phi_{i}\Big)
\end{equation}
from D, leading to
\begin{equation}
    Y(D^2\phi_{i})-Y[\phi^{\dagger}(\partial_{\mu}D^{\mu}\phi)]\phi_{i}-4Y(D^2\phi_{i})+4Y[\phi^{\dagger}(\partial_{\mu}D^{\mu}\phi)]\phi_{i}=0,
\end{equation}
which confirms that the topological chiral solutions still remain valid in the deformed system.

Recall that the energy-momentum tensor of the deformed system in Eq. \eqref{deformedenergymomentumtensor} is
\begin{equation}
    T^{\mu}_{\ \nu}=\frac{2\gamma X^{\mu\theta}X_{\theta\nu}+(1-2\gamma Y)X^{\mu}_{\ \nu}}{\sqrt{1-2\gamma Y+2\gamma^2(Y^2-X)}}-\delta^{\mu}_{\nu}\frac{1-\sqrt{1-2\gamma Y+2\gamma^2(Y^2-X)}}{2\gamma}.
\end{equation}
From the identities previously listed in Eq. \eqref{instantoncase}, $X_{00}=X_{11}=0$ immediately implies
\begin{equation}
    T_{01}=T_{10}=0.
\end{equation}
For the diagonal components, we additionally have
\begin{equation}
    \sqrt{1-2\gamma Y+2\gamma^2(Y^2-X)}=1-\gamma Y,
\end{equation}
which causes the two terms in the expression for $T^{\mu}_{\ \nu}$ to cancel. This yields
\begin{equation}
    T_{00}=T_{11}=0,
\end{equation}
confirming the vanishing of the energy-momentum tensor.

\section{Matching the $\sqrt{T\ov T}$ deformed Lagrangian of \cite{Borsato:2022tmu} via the AFSM}\label{ap:rootlag}

Here we briefly show how to derive the $\sqrt{T\ov T}$ deformed Lagrangian of \cite{Borsato:2022tmu} by integrating out the auxiliary field $\vt {}$. This appendix should be viewed as a consistency check of the auxiliary field formalism. Let us set the interaction function to $E= A\sqrt {\nu _2}$ \cite{Ferko:2024ali}, the Lagrangian is 
\begin{equation}\label{eq:Loriginal}
    \cL = \frac 12 \tr (\jt+\jt-)+ \tr (\vt{+}\vt{-}) + \tr(\jt+\vt{-}+\jt-\vt{+}) + A\sqrt {\nu _2}.
\end{equation}
The $\vt{}$ equations of motion take the form  
\begin{equation}\label{eq:deftr1}
    \begin{split}
            \tr(\vt{+}\vt{-}) +\tr(\jt+\vt{-})&=- A\sqrt {\nu _2},\\
            \tr(\vt{+}\vt{-}) +\tr(\jt-\vt{+})&=- A\sqrt {\nu _2}.
    \end{split}
\end{equation}
Which imply
\begin{equation}
    \tr (\jt+\vt{-})=\tr (\jt-\vt{+}).    
\end{equation}
It is convenient to introduce the quantity 
\begin{equation}
    \cL _0=\frac 12 \tr (\jt+\jt-)+ \tr (\vt+\vt-) + 2\tr(\jt+\vt-),
\end{equation}
we will now use \eqref{eq:deftr1} and the $\vt{}$ equations of motion repeatedly. One can start by eliminating the $\tr (\vt+\vt-)$ term by substituting \eqref{eq:deftr1}
\begin{equation}\label{eq:Lodef}
    \cL_0 - \frac 12 \tr (\jt+\jt-)-\tr (\jt+\vt-)= - A \sqrt {\nu _2}.
\end{equation}
 Now substitute the $\vt-$ equation into \eqref{eq:Lodef},
 \begin{equation}
    \begin{split}
        \cL_0 + \frac 12 \tr (\jt+\jt-)+\frac A{\sqrt {\nu _2}}\,\tr ((\vt-)^2 ) \tr (\jt+\vt+)=-A\sqrt {\nu _2},
    \end{split}
\end{equation}
and yet again the $\vt+$ equation in the $\tr (\jt+\vt+)$ term
\begin{equation}
    \begin{split}
        \cL_0 + \frac 12 \tr (\jt+\jt-)-\frac A{\sqrt {\nu _2}}\,\tr ((\vt-)^2 ) \tr (\jt+\jt+) -A^2\, \tr (\jt+\vt-)&=-A\sqrt {\nu _2}.
    \end{split}
\end{equation}
Finally, it is useful to eliminate $\tr (\jt+\vt-)$ using again \eqref{eq:Lodef}. We end up with the relation 
\begin{equation}
    \begin{split}
                (1-A^2)\,\cL_0+ \frac 12 (1+A^2)\,\tr (\jt+\jt-) + A(1-A^2)\, \sqrt {\nu _2}&= \frac A{\sqrt {\nu _2}}\,\tr ((\vt-)^2 ) \tr((\jt+)^2).
    \end{split}
\end{equation}
It is clear that one can run this sequence of steps upon flipping all signs, hence we end up with two equations 
\begin{equation}\label{eq:solvefornu}
    \begin{split}
        (1-A^2)\,\cL_0+ \frac 12 (1+A^2)\,\tr (\jt+\jt-) + A(1-A^2)\, \sqrt {\nu _2}&= \frac A{\sqrt {\nu _2}}\,\tr ((\vt-)^2 ) \tr((\jt+)^2),\\
        (1-A^2)\,\cL_0+ \frac 12 (1+A^2)\,\tr (\jt+\jt-) + A(1-A^2)\, \sqrt {\nu _2}&= \frac A{\sqrt {\nu _2}}\,\tr ((\vt+)^2 ) \tr((\jt-)^2).
    \end{split}
\end{equation}
At this point, introduce the analogue of $\nu _2$ where $v\longleftrightarrow j$, namely 
\begin{equation}
    \rho _2= \tr ((\jt+)^2)\tr ((\jt-)^2),
\end{equation}
 multiplying the two lines of \eqref{eq:solvefornu} together and setting $A= \tanh (\frac \gamma 2)$ gives \eqref{eq:Lj}, and recovers the result of \cite{Borsato:2022tmu}
\begin{equation}
    \begin{split}
\cL = -\frac 12 \left( \frac {1+A^2}{1-A^2}\right)\tr (\jt+\jt-)+ \frac A{1-A^2}\sqrt {\rho _2}.
    \end{split}
\end{equation}

 \bibliographystyle{utphys}
\bibliography{master}
\end{document}